\documentclass[twocolumn,notitlepage,floatfix,superscriptaddress]{revtex4-2}
\usepackage{amsmath, amsfonts, amsthm, amssymb, bbm}
\usepackage[margin=1in]{geometry}
\usepackage{algorithm,algpseudocode}
\usepackage{graphicx,xcolor}
\usepackage{hyperref}
\usepackage{physics}
\usepackage{qcircuit}

\newcommand{\bb}{\mathbb}
\newcommand{\mc}{\mathcal}
\newcommand{\code}[1]{[\![#1]\!]}
\DeclareMathOperator{\supp}{supp}
\newcommand{\unit}[1]{\,\mathrm{#1}}

\algblockdefx[ForEach]{ForEach}{EndForEach}[1]{\textbf{for each} #1 \textbf{do}}{\textbf{end for each}}

\newtheorem{theorem}{Theorem}
\newtheorem{lemma}[theorem]{Lemma}

\newcommand{\measureX}[1]{*+[F-:<.9em>]{#1}}

\begin{document}

\title{Fault-tolerant quantum architectures based on erasure qubits}
\author{Shouzhen Gu}
\affiliation{California Institute of Technology, Pasadena, CA 91125, USA}
\affiliation{AWS Center for Quantum Computing, Pasadena, CA 91125, USA}
\author{Alex Retzker}
\affiliation{AWS Center for Quantum Computing, Pasadena, CA 91125, USA}
\affiliation{Racah Institute of Physics, The Hebrew University of Jerusalem, Jerusalem 91904, Givat Ram, Israel}
\author{Aleksander Kubica}
\thanks{akubica@caltech.edu}
\affiliation{AWS Center for Quantum Computing, Pasadena, CA 91125, USA}
\affiliation{California Institute of Technology, Pasadena, CA 91125, USA}
\begin{abstract}
The overhead of quantum error correction (QEC) poses a major bottleneck for realizing fault-tolerant computation.
To reduce this overhead, we exploit the idea of erasure qubits, relying on an efficient conversion of the dominant noise into erasures at known locations.
We start by introducing a formalism for QEC schemes with erasure qubits and express the corresponding decoding problem as a matching problem.  
Then, we propose and optimize QEC schemes based on erasure qubits and the recently-introduced Floquet codes.
Our schemes are well-suited for superconducting circuits, being compatible with planar layouts.
We numerically estimate the memory thresholds for the circuit noise model that includes spreading (via entangling operations) and imperfect detection of erasures.
Our results demonstrate that, despite being slightly more complex, QEC schemes based on erasure qubits can significantly outperform standard approaches. 
\end{abstract}

\maketitle

\section{Introduction}

Quantum error correction (QEC) and fault-tolerant protocols~\cite{Shor1995,Steane1996,Shor1996} can benefit significantly from an ingenious design of qubits with tailored, often hardware-dependent, noise structure.
For instance, a bosonic cat qubit~\cite{Cochrane1999,Mirrahimi2014,Ofek2016,Puri2019,Guillaud2019} exhibits biased Pauli noise.
Such noise can be readily exploited, e.g., by an appropriate variant of the surface code~\cite{Tuckett2018,Tuckett2019,Tuckett2020,Bonilla2021,Dua2022,Xu2023,Higgott2023}, resulting in greatly increased QEC thresholds and reduced qubit overheads of the associated QEC protocols.

Recently, another type of qubit, often referred to as an \emph{erasure qubit}, has received significant attention.
Several theoretical proposals have described how the erasure qubit can be straightforwardly realized with, e.g., neutral atoms~\cite{Wu2022}, trapped ions~\cite{Kang2023} or superconducting circuits~\cite{Kubica2023,Teoh2023}, as well as several promising proof-of-principle experimental demonstrations~\cite{Ma2023,Scholl2023,Chou2023,Levine2023,Koottandavida2023}.
The idea behind the erasure qubit is to engineer a qubit in such a way that its dominant noise is detectable erasures~\cite{Grassl1997}.
Importantly, the knowledge of locations of erasures can be efficiently leveraged by QEC protocols (that may be based on the surface code) and decoding algorithms, leading to high QEC thresholds and an improved subthreshold scaling of the logical error rate~\cite{Stace2009,Delfosse2021,Sahay2023}.

For QEC protocols to benefit from erasure qubits, the following requirements have to be satisfied:
(i) a large erasure bias, defined as the ratio of the probability of an erasure to the probability of any other residual error within the computational subspace,
(ii) an implementation of standard quantum circuit operations, including state preparation, unitary gates and readout, in a way that preserves an erasure bias, and
(iii) the ability to perform an erasure check capable of reliably detecting erasures without introducing additional errors within the computational subspace.
Thus, it appears that approaches based on erasure qubits, compared to the standard ones, are more challenging to implement.
However, while there is an additional cost associated with, e.g., a careful design of the erasure check, once these additional building blocks are available, then new possibilities for optimized QEC protocols and fault-tolerant architectures open up.

In this article, we address the question of designing fault-tolerant architectures that are optimized for and make full use of erasure qubits.
We start by introducing a formalism for QEC protocols with erasure qubits and express the corresponding decoding problem as the hypergraph matching problem.
This, in turn, allows us to design decoding algorithms, which, in many relevant scenarios, may be based on the matching algorithm~\cite{Dennis2002,Higgott2023}.
We then focus on fault-tolerant architectures and find that erasure qubits are particularly suitable for a recently-introduced family of QEC codes, Floquet codes~\cite{Hastings2021,Haah2022}.
In particular, in one realization of erasure qubits via the dual-rail encoding~\cite{Duan2010}, the minimal set of quantum circuit operations necessary to implement Floquet codes consists of state preparation, readout, single-qubit gates and erasure checks (which also play the role of entangling operations).
We discuss possible physical implementations of erasure checks in the context of superconducting circuits.
To benchmark our scheme, we numerically estimate the memory thresholds of Floquet codes against circuit noise comprising erasures, Pauli errors and measurement errors (for readout and erasure checks).
Lastly, we analyze further optimizations of Floquet codes that lead to the reduced qubit overhead, as well as find the smallest Floquet codes with distance two and four, which require, respectively, 4 and 16 qubits.

\section{QEC protocols with erasure qubits}

Analyzing QEC protocols with erasure qubits poses some challenges.
First, modelling each erasure qubit may require at least a three-level system.
Second, quantum circuit operations might, in principle, introduce correlated coherent errors in the presence of erasures.
Consequently, there seems to be little hope for efficient simulation methods akin to the ones used for stabilizer circuits~\cite{Gottesman1998,Aaronson2004}.
In addition, decoding algorithms do not typically consider erasures and their spread.

In this section, we describe how making certain simplifying assumptions allows us to efficiently decode and simulate QEC protocols with erasure qubits; see Appendix~\ref{sec_formal} for a detailed description of our formalism.
In particular, we phrase the decoding problem as the hypergraph matching problem.

We emphasize that our formalism and numerical simulations go beyond the paradigm of erasure qubits.
Namely, they can be used for QEC schemes with leakage, allowing us to quantify the potential gains from the ability to detect leakage~\cite{AliferisTerhalleakage,Varbanov2020,Miao2023}.

\subsection{Setting the formalism}
\label{sec_formalism}

To describe QEC protocols, we use quantum circuits that consist of the following standard single-qubit (1Q) and two-qubit (2Q) operations:
(i) 1Q state preparation (of eigenstates of Pauli operators),
(ii) 1Q readout (in any Pauli basis),
(iii) 1Q Clifford gates, and
(iv) 2Q controlled-Pauli $CP$ gates, where $P \in \{X,Y,Z\}$,
as well as the additional operations:
(v) 1Q erasure checks,
(vi) 1Q reset (of the erasure qubit),
(vii) 2Q projective measurements of Pauli $PP$ operators.
We refer to such circuits as \emph{erasure circuits}.
In contrast, \emph{stabilizer circuits} consist only of operations (i)-(iv) and (vii). 
An example of an erasure circuit is presented in Fig.~\ref{fig_circuits}.

\begin{figure}[ht!]
\centering
(a)\qquad\quad
$\Qcircuit @C=.8em @R=.7em {
\lstick{\ldots} & \gate{P} & \gate{\mathrm{EC^*}} & \qw & \lstick{\ldots}  & \\
\lstick{\ket +} & \ctrl{-1} & \gate{\mathrm{EC^*}} & \ctrl{1} & \meter_{\qquad\qquad\quad\ X} &\\
& \lstick{\ldots} & \gate{\mathrm{EC^*}} & \gate{P} & \qw & \lstick{\ldots\!\!}}$
\quad\\
\vspace*{10pt}
(b)\qquad\quad 
$\Qcircuit @C=.8em @R=.7em {
\lstick{\ldots} & \measure{\mc E} & \gate{P} & \measure{\mc E} & \gate{\mathrm{EC^*}} & \measure{\mc E} & \qw & \lstick{\ldots}  & & \\
\lstick{\ket +} & \measure{\mc E} &\ctrl{-1} & \measure{\mc E} & \gate{\mathrm{EC^*}} & \measure{\mc E} & \ctrl{1} & \measure{\mc E} & \meter_{\qquad\qquad\quad\ X} &\\
& & \lstick{\ldots} & \measure{\mc E} & \gate{\mathrm{EC^*}} & \measure{\mc E} & \gate{P} & \measure{\mc E} & \qw & \lstick{\ldots\!\!}}$
\caption{(a) An erasure circuit consists of standard operations, such as state preparation, readout and controlled-Pauli gates, as well as erasure checks with reset (denoted by $\mathrm{EC^*}$).
(b) The same quantum circuit with all possible erasure locations (denoted by $\mc E$) explicitly inserted.
This circuit is used in the ancilla scheme for Floquet codes.
}
\label{fig_circuits}
\end{figure}

To simulate QEC protocols, we first need to explicitly include all the erasure locations (where erasures may happen) into erasure circuits and then replace each ideal operation (i)-(vii) by its noisy counterpart.
We realize the latter by following a standard procedure of adding appropriate Pauli noise $\mc P$ after each operation and adding bit-flip noise $\mc N$ to each outcome (for readout, erasure checks and projective measurements), as depicted in Table~\ref{tab_ideal2simulated}.
To realize the former, we choose to insert one erasure location on each wire between each two consecutive operations (i)-(vii); see Fig.~\ref{fig_circuits}(b) for an illustration.
The noise strengths could be arbitrary at every location, however, for the purposes of our analysis, we describe the noise by three parameters: $e$, the erasure rate; $p$, the Pauli error rate; and $q$, the classical bit-flip noise rate.

\begin{table}[t]
\begin{tabular}{| p{20mm} | c | c |}
 \hline
 {\bf operation} & {\bf ideal} & {\bf simulated} \\
 \hline\hline
 state \mbox{preparation} & \Qcircuit @C=.8em @R=.7em {
        \lstick{\ket \psi} & \qw} &
\Qcircuit @C=.8em @R=.7em {
        \lstick{\ket \psi} & \measure{\mc P\!\left(\frac{3}{2}p\right)} & \qw
    } \\
 \hline
readout & \hspace*{-5mm}\Qcircuit @C=.8em @R=.7em {
        & \meter_{\qquad\qquad\ P}
    } &
\Qcircuit @C=.8em @R=.7em {
& \gate{\mathrm{EC}} & \meter_{\qquad\qquad\qquad\ P} &\\
& \measureX{\mc N(q)} \cwx & \cw & \cw \\
& & \measureX{\mc N(q)} \cwx[-2] & \cw }\raisebox{-18.3mm}{ } \\
 \hline
 erasure check with reset & \Qcircuit @C=.8em @R=.7em {
        & \gate{\mathrm{EC^*}} & \qw
    } &
\Qcircuit @C=.8em @R=.7em {
& \gate{\mathrm{EC}} & \gate{\mathrm R}& \measure{\mathcal P(p)} & \qw \\
& \measureX{\mathcal N(q)} \cwx & \cw & \cw & \cw}\raisebox{-11mm}{ } \\
 \hline
entangling gate & \Qcircuit @C=.8em @R=.7em {
        & \multigate{1}{G} & \qw\\
        & \ghost{G} & \qw}
        & \Qcircuit @C=.8em @R=.7em {
         & \multigate{1}{G} & \multimeasure{1}{\mc P(p)} & \qw\\
        & \ghost{G} & \ghost{\mc P(p)} & \qw}\raisebox{-8mm}{ } \\
 \hline
projective measurement & \Qcircuit @C=.8em @R=.7em {
        & \multigate{1}{\Pi_{PP}} & \qw \\
        & \ghost{\Pi_{PP}} & \qw
    } & \Qcircuit @C=.8em @R=.7em {
        & \multigate{1}{\Pi_{PP}} & \multimeasure{1}{\mathcal P(p)} & \qw \\
        & \ghost{\Pi_{PP}} & \ghost{\mathcal P(p)} & \qw \\
        & \measureX{\mathcal N(q)} \cwx & \cw & \cw}\raisebox{-15.4mm}{ } \\
 \hline
\end{tabular}
\caption{
Mapping of an ideal circuit to a simulated circuit.
Note that erasures occur at the locations in between ideal operations; see Fig.~\ref{fig_circuits}(b).
For concreteness, we choose $\mc P(p)$ to be the 1Q or 2Q depolarizing channel (determined by its support) with error rate $p$, and $\mc N(q)$ to be the binary symmetric channel (that flips the measurement outcome) with error rate $q$.
In general, $\mc P$ and $\mc N$ can represent arbitrary Pauli and binary channels.
}
\label{tab_ideal2simulated}
\end{table}

In contrast to standard QEC protocols, with erasure qubits not only do we have Pauli noise that affects quantum circuits but also erasures.
Erasures are probabilistic processes happening at erasure locations and taking the state from the computational subspace to some orthogonal subspace, which we refer to as the erasure subspace.
Furthermore, erasures can spread (probabilistically) via 2Q operations.
We envision the following two scenarios.
\begin{itemize}
    \item Erasure-erasure spread: an erasure spreads to another erasure.
    \item Erasure-Pauli spread: an erasure spreads to a Pauli error.
\end{itemize}
One concrete realization of the erasure-Pauli spread, which we refer to as the erasure-depolarization spread, is when an erasure spreads to either Pauli $X$, or $Y$, or $Z$, each with probability $1/4$.
In other words, any qubit affected by an erasure causes full depolarization of any other qubit that is involved in the same 2Q operation.
In the rest of the article we focus on the erasure-depolarization spread.

The key part of QEC protocols with erasure qubits is the ability to detect erasures.
Erasure detection can be achieved by either readout or erasure checks.
Each erasure check performs a nondestructive measurement that distinguishes the states in the computational and erasure subspaces.
For simplicity, in our numerical simulations in Sec.~\ref{sec:simulationresults} and Sec.~\ref{sec_smallest} we assume that each erasure check is immediately followed by reset of the erasure qubit that reinitializes it in the computational subspace (and do not include an erasure location in between the erasure check and reset).
We envision the following two scenarios.
\begin{itemize}
    \item Conditional (active) reset that depends on the outcome of the erasure check (and possibly other previous erasure checks).
    \item Unconditional (passive) reset that is independent of any erasure check outcome.
\end{itemize}
For concreteness, in either case we assume that the erasure qubit is reinitialized in the maximally mixed state in the computational subspace.
We further assume that reset acts trivially on the computational subspace.
In the rest of the article we focus on unconditional reset.

We emphasize that in our formalism we assume that the operations (i)-(vii) do not create coherences between the computational and erasure subspaces.
This, in turn, allows us to efficiently sample from erasure circuits; see Sec.~\ref{sec:sampling} for details.
Also, our assumption about the erasure-depolarization spread allows us to push erasures through entangling operations and model entangling operations as always acting on the computational subspace.
We leverage this observation to express the decoding problem for erasure circuits as the hypergraph matching problem; see Sec.~\ref{sec:decode} for details.

Our approach generalizes straightforwardly to leakage simulations as long as there are no coherences between the computational and leakage subspaces (regardless of the number of leaked levels), the only difference being the lack of erasure check operations; see Fig.~\ref{fig:mainresults}(b)(e) and Fig.~\ref{fig:crosssectionssupplementary}.
When there are coherences, the simulation may become less efficient because we might not be able to use the stabilizer formalism to describe the relevant states.

\subsection{Decoding problem for erasure circuits}
\label{sec:decode}

Let us first describe the decoding problem for stabilizer circuits.
In addition to the stabilizer circuit, we also need to specify a distribution of Pauli errors that are placed at \emph{spacetime locations} between operations of the circuit and a set of \emph{detectors} $\{V_i\}$.
By definition, a detector is a product of measurement outcomes in the circuit that is deterministic in the absence of errors and gives information about possible errors when triggered.
The stabilizer circuit may describe the implementation of, e.g., a stabilizer code, where detectors are products of consecutive stabilizer measurements, or a Floquet code, where detectors are more complicated.
Given detectors which are triggered after running the circuit, the decoding problem is to find a Pauli recovery which undoes the errors that occurred.

The decoding problem for stabilizer circuits can be phrased as a hypergraph matching problem.
This formulation is particularly useful when the distribution of Pauli errors is either equal to or approximated by a product distribution of binary random variables, often referred to as \emph{error mechanisms}.
By definition, an error mechanism is a pair $(P_i,p_i)$ such that the Pauli error $P_i$ is inserted at specified spacetime locations in the circuit with probability $p_i$.
When $P_i$ occurs, it causes a subset of detectors $\mc T_i\subseteq \{V_i\}$ to be triggered.
If we then define a weighted hypergraph $H=(\{V_i\}, \{\mc T_i\})$, where each hyperedge $\mc T_i$ has weight $w(\mc T_i) = \log((1-p_i)/p_i)$,
then the problem of finding the most likely error triggering a subset of detectors is equivalent to the minimum-weight hypergraph matching problem on $H$, i.e., for a given subset of vertices $\nu\subseteq \{V_i\}$, find a subset of hyperedges $\tau\subseteq \{\mc T_i\}$ with lowest total weight $\sum_{\mc T\in\tau} w(\mc T)$, such that $\bigoplus_{\mc T\in\tau} \mc T = \nu$, where $\oplus$ denotes the symmetric difference of sets.
The recovery operator is the product of all Pauli errors (propagated to the end of the circuit) that correspond to the hyperedges in $\tau$.

The decoding problem for \emph{erasure circuits} is still to find a Pauli recovery.
This time, in addition to the triggered detectors, we also have erasure check outcomes.
In what follows, we describe a mapping of erasure circuits to stabilizer circuits, where erasures are converted into independent Pauli error mechanisms.
This mapping, in turn, allows us to phrase the decoding problem for erasure circuits as a hypergraph matching problem.

We decompose the erasure circuit $C_E$ into \emph{segments}.
By definition, a segment $s$ is the worldline of a single qubit $q$ between two consecutive reset operations~\footnote{For brevity, we also refer to state preparation and readout as reset operations.}; see Fig.~\ref{fig:circuits_decoder}(a) for an illustration.
We also define the entangling operations of $s$ as those with nontrivial support on $s$ and the spacetime locations associated with $s$ as those immediately following the entangling operations of $s$.
To map $C_E$ to a stabilizer circuit with independent Pauli error mechanisms, we modify each segment $s$ of $C_E$ by removing the erasure checks and reset operations in $s$ and add appropriate error mechanisms at the locations associated with $s$.
This mapping is guaranteed by Lemma~\ref{thm:segmentmapping}.

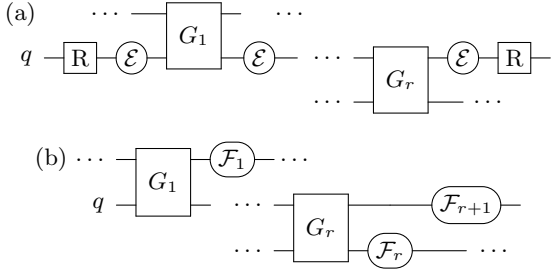
\begin{figure}[t!]
\centering
(a) $\Qcircuit @C=.8em @R=.7em {
&  & \lstick{\ldots} & \multigate{1}{G_1} & \ghost{I} & \lstick{\ldots}\\
\lstick{q} & \gate{\mathrm R} & \measure{\mc E} & \ghost{G_1} & \measure{\mc E} & \ghost{I} & \lstick{\ldots} & \multigate{1}{G_{r}} & \measure{\mc E} & \gate{\mathrm{R}} & \qw \\
& & & & & & \lstick{\ldots} & \ghost{G_{r}} & \qw & \lstick{\ldots}}$\\
\vspace*{10pt}
(b)\qquad $\Qcircuit @C=.8em @R=.4em {
\lstick{\ldots} & \multigate{1}{G_1} & \measure{\mc F_1} & \qw & \lstick{\ldots}\\
\lstick{q} & \ghost{G_1} & \ghost{G} & \lstick{\ldots} & \multigate{1}{G_{r}} & \qw & \measure{\mc F_{r+1}} & \qw\\
& & & \lstick{\ldots} & \ghost{G_{r}} & \measure{\mc F_{r}} & \qw & \lstick{\ldots}}$
\caption{
(a) A segment $s$ of the qubit $q$ in the erasure circuit $C_E$ with entangling operations $G_i$ and erasure locations $\mc E$.
(b) A stabilizer circuit $C$ with spacetime locations $\mc F_i$.
By placing spacetime correlated Pauli errors at $\overline{\mc F_i}$ with appropriate probabilities, $C$ becomes equivalent to $C_E$; see Algorithm~\ref{alg:erasuretostabilizercircuit} for details.}
\label{fig:circuits_decoder}
\end{figure}

\begin{lemma}
    \label{thm:segmentmapping}
    Let $s$ be a segment of an erasure circuit $C_E$.
    Given the outcomes $\vec{d}$ of erasure checks in $s$, the distribution of errors introduced by erasures in $s$ is equivalent to a distribution of spacetime correlated Pauli errors $\mc P$  that can be described by independent error mechanisms $\{(P_{i,j},p_i)\}$.
\end{lemma}

\begin{proof}
The proof proceeds in three steps.
First, we find the distribution of Pauli errors caused by erasures in the segment.
This distribution can be described by disjoint events which are correlated depolarizing channels applied at different spacetime locations in $C_E$, caused by the erasure-depolarization spread.
Second, we show that this distribution can also be described by a product of \emph{independent} events which are spacetime correlated depolarizing channels.
Third, we decompose each of these spacetime correlated depolarizing channels into independent error mechanisms.

We start by introducing some notation.
Let $G_i$ be the $i$-th entangling operation in $s$ and $\mc F_i$ be the spacetime location associated with $s$ placed after $G_i$ defined via $\supp{\mc F_i} = \supp{G_i}\setminus \{ q\}$, where $i\in \{1,\ldots, r\}$ and $q$ labels the qubit identified with $s$; see Fig.~\ref{fig:circuits_decoder}.
For convenience, we use $G_0$ and $G_{r+1}$ to denote the first and second reset operations in $s$, and define $\mc F_{r+1}$ to be the location at $q$ after the second reset.
In the case when $G_i$ is a 2Q projective measurement, we imagine the classical bit containing the outcome to be a qubit and $\mc F_i$ to include that qubit.
We also define
\begin{equation}
\overline{\mc F_i} = \bigcup_{j=i}^{r+1} \mc F_j\, .
\end{equation}

We want to find the distribution of Pauli errors $\mc P$ caused by erasures in $s$.
Let $A_i$ be the event that it was first erased at any location in between $G_{i-1}$ and $G_i$, where $i\in\{1,\ldots,r+1\}$.
When $A_i$ occurs, it causes all qubits connected to $q$ through subsequent entangling operations to be fully depolarized, i.e., fully depolarizing channels are added at spacetime locations $\overline{\mc F_i}$.
Note that $\{A_i\}$ are disjoint events.
Given the erasure check outcomes $\vec{d}$ (whose probability distribution depends on the erasure probabilities and the false positive and negative detection rates of the erasure checks in the segment $s$), we can calculate the posterior probabilities
\begin{equation}
    a_i = \Pr\left(A_i\middle|\vec{d}\right)\, .
\end{equation}
We then obtain the distribution of Pauli errors $\mc P$ by sampling disjoint events with probability $a_i$ and inserting fully depolarizing channels at $\overline{\mc F_i}$ whenever the corresponding event is sampled.

In the description of $\mc P$, we have disjoint events
rather than independent ones.
To obtain the desired description, we show that $\mc P$ can also be obtained by sampling $r+1$ independent events $\{B_i\}$, where $B_i$ is defined as a binary random variable with probability 
\begin{equation}
    b_{i} = a_{i}\prod_{j=1}^{i-1}(1-b_j)^{-1}\, ,
\end{equation}
and placing fully depolarizing channels at spacetime locations $\overline{\mc F_i}$ whenever $B_i$ is sampled.
To do that, observe that for $i<j$ a composition of fully depolarizing channels at $\overline{\mc F_i}$ and $\overline{\mc F_j}$ is equivalent to the fully depolarizing channels at $\overline{\mc F_i}$, since $\overline{\mc F_i}\supseteq\overline{\mc F_j}$.
Therefore, fully depolarizing channels are placed exclusively at $\overline{\mc F_i}$ iff we sample $B_i$ but no other $B_j$ for $j<i$, which happens with probability
\begin{equation}
b_i \prod_{j=1}^{i-1} (1-b_j) = a_{i} = \Pr\left(A_{i}\middle|\Vec{d}\right)\, .
\end{equation}

Finally, we further decompose the depolarizing channels resulting from the events $\{B_i\}$ into independent error mechanisms.
Namely, for each $B_i$ we introduce $4^{\left|\overline{\mc F_i}\right|}-1$ error mechanisms, each corresponding to a different nontrivial Pauli error $P_{i,j}$ that can be placed at spacetime locations $\overline{\mc F_i}$  with probability
\begin{equation}
    p_i = \tfrac 1 2 - \tfrac 1 2 \left(1 - b_i\right)^{2^{1-2\left|\overline{\mc F_i}\right|}}\, .
\end{equation}
The resulting product distribution of independent error mechanisms $\{(P_{i,j},p_i)\}$ is equivalent to $\mc P$.
To show that, we use the same reasoning as in Claim 1 of Ref.~\cite{Chao2020surfacecode}.
The analysis there considered Pauli errors on $m$ qubits which occur at the same time, but the only necessary ingredient is that the set of possible errors under composition forms a group isomorphic to $\bb Z_2^{2m}$ which is true in our case with $m=|\overline{\mc F_i}|$.
This concludes the proof.
\end{proof}

For the reader's convenience, we recap the conversion of an erasure circuit into a stabilizer circuit with independent error mechanisms in Algorithm~\ref{alg:erasuretostabilizercircuit}.

\linespread{1.12}\selectfont
\begin{algorithm}[H]
    \caption{Conversion of an erasure circuit to a stabilizer circuit with independent error mechanisms}
    \label{alg:erasuretostabilizercircuit}
    \textbf{Input:}\\
    erasure circuit $C_E$, erasure check outcomes $\Vec{d}$\\
    \textbf{Output:} \\stabilizer circuit $C$, error mechanisms $\{(P_{i,j},p_i)\}$
    \begin{algorithmic}[1]
        \State $S \gets \{\text{segments in }C_E\}$
        \ForEach{$s\in S$}
            \State $\{\mc F_i\} \gets $ spacetime locations associated with $s$
            \ForEach{$i$}
                \State $a_i \gets \Pr\left(A_i \middle| \vec{d}\right)$
                \State $b_i \gets a_{i}\prod_{j=1}^{i-1}(1-b_j)^{-1}$
                \State $p_i \gets \frac 1 2 - \frac 1 2 \left(1 - b_i\right)^{2^{1-2\left|\overline{\mc F_i}\right|}}$
                \ForEach{nontrivial Pauli error $P_{i,j}$ at $\overline{\mc F_i}$}
                    \State include error mechanism $(P_{i,j},p_i)$
                \EndForEach
            \EndForEach
        \EndForEach
        \State $C \gets \text{$C_E$ with deleted erasure checks and reset}$
        \State \Return $C$, $\{(P_{i,j},p_i)\}$
    \end{algorithmic}
\end{algorithm}
\linespread{1}\selectfont

We finish this subsection with a few remarks.
\begin{itemize}
    \item Depolarizing noise can be decomposed into independent error mechanisms~\cite{Chao2020surfacecode}.
    For arbitrary Pauli channels an exact decomposition may not exist, but one may use an approximate one~\cite{delfosse2023splitting}.
    \item In many scenarios, e.g., the surface code or Floquet codes, the distribution of Pauli errors can be further approximated by a product of independent error mechanisms that trigger at most two detectors.
    Consequently, the decoding problem reduces to the graph matching problem, which is efficiently solvable~\cite{Edmonds1965}.
    \item The number of error mechanisms added for each segment $s$ in Algorithm~\ref{alg:erasuretostabilizercircuit} is exponential in the length of $s$.
    If reset operations occur at constant time intervals in the erasure circuit $C_E$, then the total number of added error mechanisms is proportional to the size of $C_E$.
    \item If reset operations occur between every entangling operation in the erasure circuit (as is the case in our numerical simulations), then the resulting error mechanisms are not time correlated and can simply be described by depolarizing channels; see Appendix~\ref{app:example}.
    \item For simplicity, we focused on the scenario of the erasure-depolarization spread.
    The analysis for the deterministic erasure-erasure spread is similar. However, the probability of erasure between entangling operations may need to be conditioned on many erasure check outcomes, not just within one segment.
\end{itemize}

\subsection{Sampling erasure circuits}
\label{sec:sampling}

We envision two ways of sampling from an erasure circuit for simulation purposes. In the first method, we use one bit of information per qubit to represent its erasure state. This allows us to efficiently simulate erasures, erasure checks and reset operations by sampling from the correct probabilities and updating this classical data. For the qubits in the computational subspace, we keep track of the stabilizers and update them after standard stabilizer circuit operations using the Gottesman-Knill theorem~\cite{Gottesman1998,Aaronson2004}. When a qubit is erased, we randomize measurement outcomes involving it and depolarize qubits that interact with it through entangling operations.

Alternatively, we may simulate the circuit by first sampling all erasure check detection events.
Conditioned on the results, we then use the results of Sec.~\ref{sec:decode} to obtain stabilizer circuits that we can sample from.
Each independent error mechanism in the resulting circuit can be simulated by adding a classical bit that determines whether or not the corresponding error is applied.
However, this is often unnecessary because if the Pauli errors occur at the same time (as is the case for us; see Appendix~\ref{app:example}), it can be simulated by depolarizing channels.
The advantage of this approach is that after obtaining the erasure detection events, the same circuit can be used for sampling and decoding.

\section{Scalable architecture with erasure qubits}

In this section, we describe how a realization of erasure qubits via the dual-rail encoding naturally leads to a fault-tolerant architecture based on Floquet codes (which we briefly overview).
Our approach is particularly well-suited for superconducting circuits.
We also discuss possible hardware implementations of physical operations needed to realize Floquet codes.

\subsection{From erasure qubits to Floquet codes}

\begin{figure}[ht]
\centering
\includegraphics[width=.9\columnwidth]{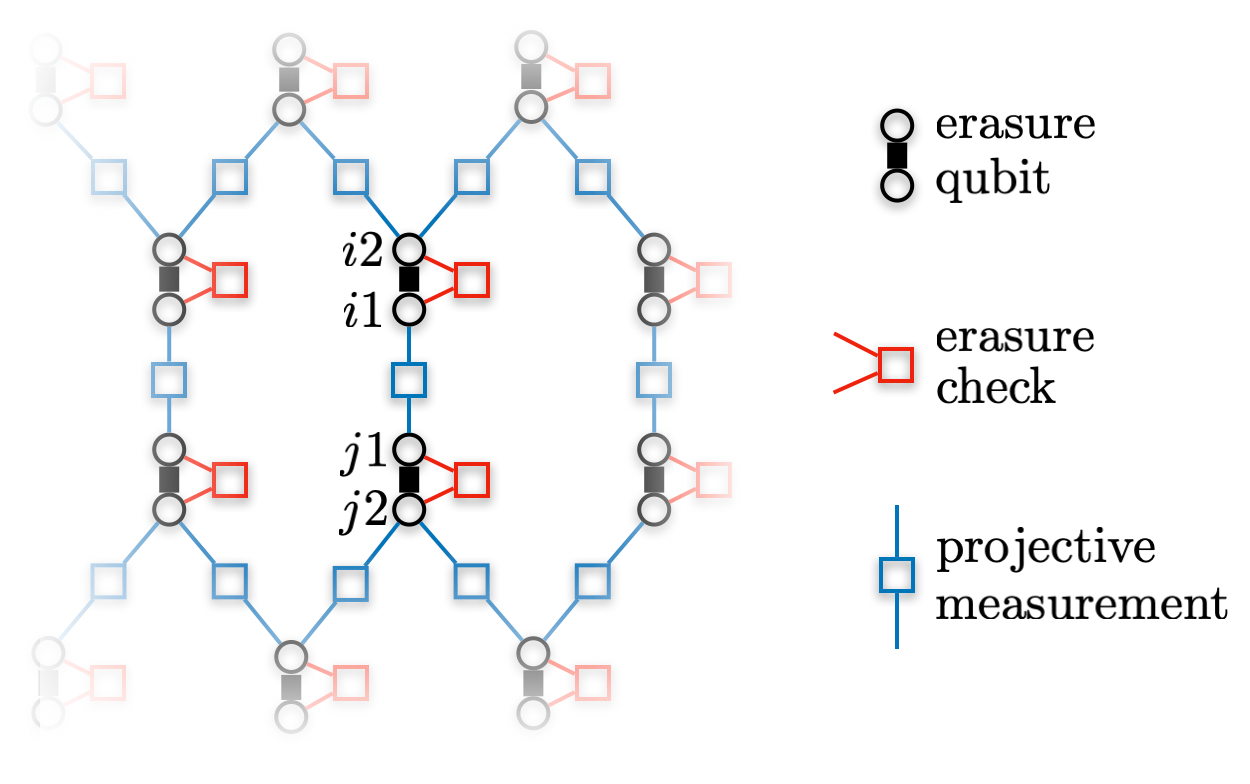}
\caption{
A planar layout of erasure qubits realized via the dual-rail encoding.
Projective measurements of Pauli $Z_{i1}Z_{i2}$ and $Z_{i1}Z_{j1}$ operators implement, respectively, an erasure check for the erasure qubit $i$ and a Pauli $\overline{Z_i Z_j}$ measurement on the computational subspace of two erasure qubits $i$ and $j$.
This layout is well-suited for Floquet codes, benefiting from their low qubit connectivity.
The erasure qubit can be realized by two coupled transmons, the erasure check via an LC element~\cite{Kubica2023} and the projective measurement via a single transmon; see Appendix~\ref{Implementation}.
}
\label{fig_physical}
\end{figure}

One of the simplest ways to realize the erasure qubit is via the dual-rail encoding
\begin{equation}
\ket{\overline 0} \mapsto \ket{01},\quad \ket{\overline 1} \mapsto \ket{10},
\end{equation}
which is particularly well-suited for superconducting circuits.
Namely, the dominant noise for this quantum computing platform is the amplitude damping noise~\cite{Yan2018,Burnett2019} that describes the energy relaxation from the excited state $\ket 1$ to the ground state $\ket 0$.
A single amplitude damping event can be detected as it maps any state of the erasure qubit to the state $\ket{00}$ which is orthogonal to the computational subspace $\mathrm{span}\{\ket{01},\ket{10}\}$.
Consequently, the effective noise afflicting the qubit is dominated by detectable erasures.
We remark that a few recent experiments demonstrated the erasure qubit via the dual-rail encoding using either two transmons~\cite{Levine2023} or two 3D cavities~\cite{Chou2023,Koottandavida2023}.

Observe that a projective measurement of a Pauli $Z_{i1}Z_{i2}$ operator, where $i1$ and $i2$ label two qubits forming the erasure qubit $i$ via the dual-rail encoding, is sufficient to implement an erasure check; see Fig.~\ref{fig_physical}.
Namely, a $+1$ measurement outcome implies that the state is outside the computational subspace $\mathrm{span}\{\ket{01},\ket{10}\}$, and the erasure qubit has suffered from an erasure.
However, a projective measurement of a Pauli $Z_{i1} Z_{j1}$ (or $Z_{i2} Z_{j2}$) operator supported on qubits from two different erasure qubits $i$ and $j$ implements a Pauli $\overline{Z_i Z_j}$ measurement on their computational subspace.
Therefore, the ability to perform projective measurements of Pauli $ZZ$ operators, together with single-qubit Hadamard $\overline H$ and phase $\overline S$ gates on the computational subspace, is sufficient to implement erasure checks and Pauli $\overline{XX}$, $\overline{YY}$, $\overline{ZZ}$ measurements on the computational subspace.
This, in turn, allows us to implement Floquet codes with erasure qubits (where we implicitly assume the capability of single-qubit state preparation and readout in the computational basis); see Fig.~\ref{fig_physical} for an illustration.

\subsection{Floquet codes}
\label{sec_floquet}

\begin{figure}[ht!]
\centering
\hspace*{-4mm}\begin{minipage}{0.6\columnwidth}
\includegraphics[width=\linewidth]{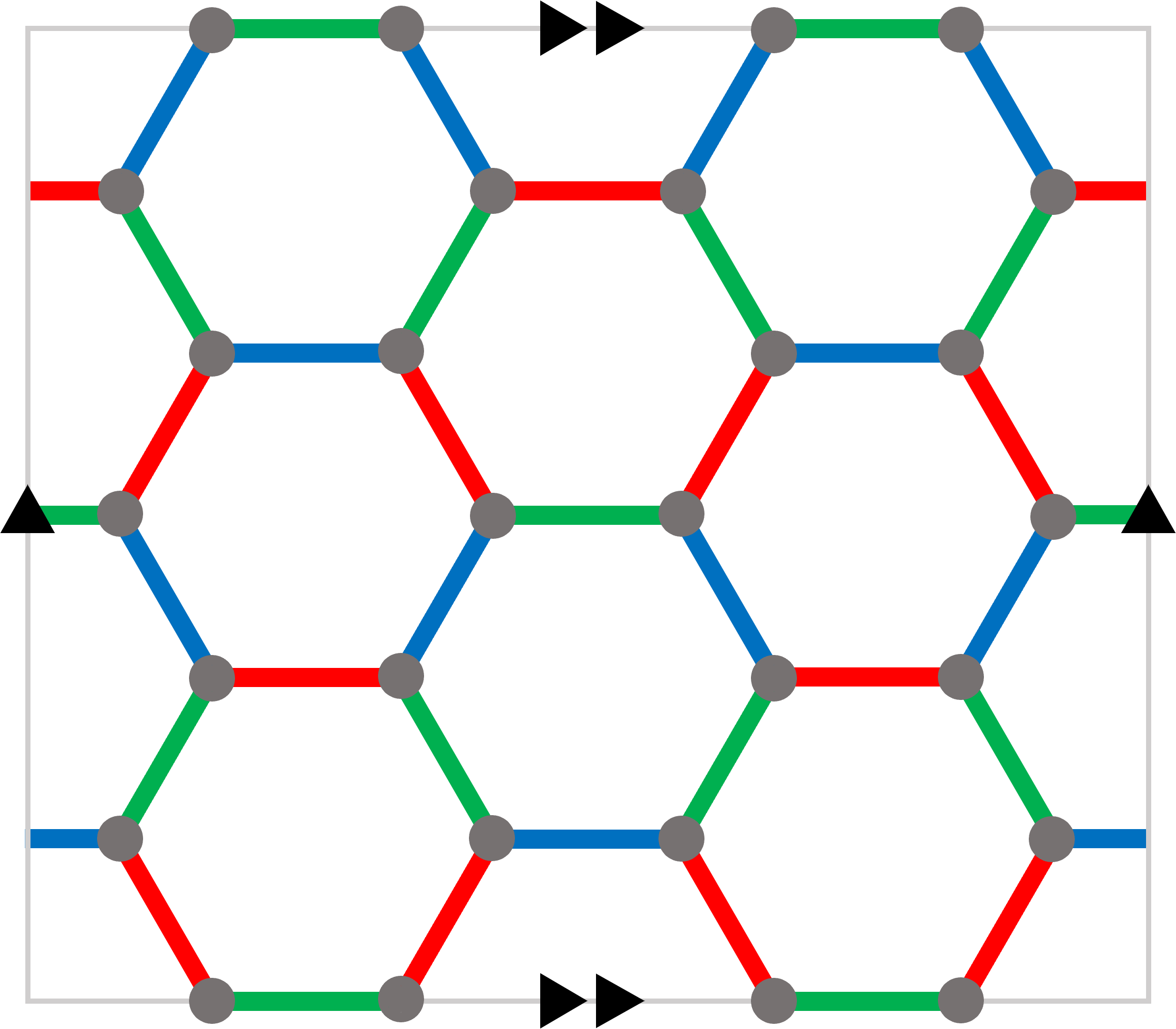}
\end{minipage}\quad
\begin{minipage}{0.3\columnwidth}
\begin{tabular}{|c|c|}
    \hline
    {\bf round} & {\bf measure}\\
    \hline\hline
    $t\equiv 0$ & $X$\includegraphics[width=0.35\columnwidth]{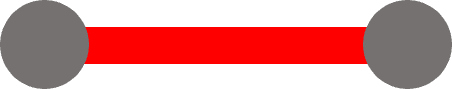}$X$\\
    \hline
    $t\equiv 1$ & $Z$\includegraphics[width=0.35\columnwidth]{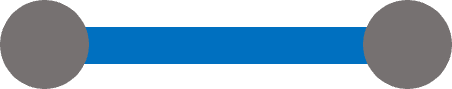}$Z$\\
    \hline
    $t\equiv 2$ & $X$\includegraphics[width=0.35\columnwidth]{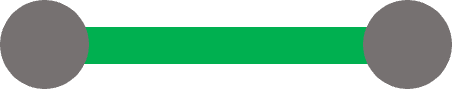}$X$\\
    \hline
    $t\equiv 3$ & $Z$\includegraphics[width=0.35\columnwidth]{figs/rededge.png}$Z$\\
    \hline
    $t\equiv 4$ & $X$\includegraphics[width=0.35\columnwidth]{figs/blueedge.png}$X$\\
    \hline
    $t\equiv 5$ & $Z$\includegraphics[width=0.35\columnwidth]{figs/greenedge.png}$Z$\\
    \hline
\end{tabular}
\end{minipage}
\caption{A (graph-based) Floquet code can be defined on a hexagonal lattice with periodic boundary conditions.
The code is realized via a sequence of measurements of two-qubit Pauli operators (depicted as red, blue and green edges) on neighboring qubits (gray dots).
A sequence of measurements described in the table gives rise to the CSS honeycomb code.
}
\label{fig_layout_st}
\end{figure}

The first and arguably simplest example of Floquet codes is the honeycomb code~\cite{Hastings2021,Haah2022}, which is defined on a hexagonal lattice with either periodic or open boundary conditions.
The honeycomb code is realized by placing qubits on the vertices $V$ and measuring two-qubit Pauli operators associated with the edges $E$ in a specified sequence.
Namely, Pauli $XX$, $YY$ and $ZZ$ operators are associated with red, blue and green edges, respectively, and are measured at a round $t \mod 3 = 0,1,2$.

One way to generalize the honeycomb code, which we refer to as a \emph{graph-based} Floquet code, is by defining a QEC code based on a connected graph $G = (V,E)$.
We require that the vertices $V$ are three-valent and the edges $E$ are three-colorable, i.e., the edges split into three sets, $E = E_0\sqcup E_1\sqcup E_2$, and no two different edges from $E_i$ are incident.
We place qubits on the vertices $V$ and consider a measurement sequence of period three, where at a round $t\mod 3 = 0,1,2$ we measure Pauli $XX$, $YY$ and $ZZ$ operators associated with edges in $E_0$, $E_1$ and $E_2$, respectively.
The definition of graph-based Floquet codes is motivated by the possibility of having a native implementation of two-qubit Pauli measurements with erasure qubits; it also guarantees low qubit connectivity.
We remark that our exhaustive search in Sec.~\ref{sec_smallest} finds the the smallest graph-based Floquet codes with distance two and four.

We can also consider a CSS version of graph-based Floquet codes, defined using a period-six measurement sequence; see Fig.~\ref{fig_layout_st} for an illustration of the CSS honeycomb code~\cite{Davydova2023}.
In what follows, we mostly focus on CSS Floquet codes, as they outperform the non-CSS counterparts; see Sec.~\ref{sec:simulationresults}.

So far, we have only discussed examples of Floquet codes without defining them.
The foundational idea behind Floquet codes is that logical information is encoded in a dynamically evolving codespace.
Consequently, a Floquet code $\mc C$ can be defined by a sequence of measurements rounds $\mc M_0, \mc M_1,\dots$, where each round $\mc M_i$ consists of a set of commuting Pauli operators.
From that perspective, Floquet codes are synonymous with a sequence of code switchings~\cite{Bombin2009,Horsman2012,Paetznick2013,Bombin2015,Kubica2015,Vuillot2019} or dynamic automorphism codes~\cite{davydova2023DAcodes}.
Note that the operators from $\mc M_i$ and $\mc M_j$ may not commute for $i\ne j$.
In each round, the codespace is stabilized by an instantaneous stabilizer group (ISG) $\mc S_i$, which is an abelian subgroup of the Pauli group not containing $-I$.
Measuring new operators in $\mc M_i$ takes the previous codespace with ISG $\mc S_{i-1}$  into a new codespace stabilized by $\mc S_i$.
The new ISG $\mc S_i$ is generated by $\mc M_i$ along with all elements of $\mc S_{i-1}$ that commute with the new measurements.
We remark that stabilizer~\cite{Gottesman1997} and subsystem~\cite{Poulin2005} codes correspond to Floquet codes with a measurement sequence of period one and two, respectively.

We can specify the code parameters of a Floquet code $\mc C$ as follows.
The ISG $\mc S_i$ can be viewed as a stabilizer code with $k_i\geq 0$ logical qubits.
The sequence $k_0,k_1,\ldots$ is nonincreasing, and therefore becomes a constant after some number of measurement rounds.
We thus define the number of logical qubits of $\mc C$ to be $k_\mc C = \lim_{i\rightarrow\infty} k_i$. 
The distance of $\mc C$ should be defined as the circuit distance, i.e., the smallest number of spacetime faults that are undetectable yet cause a logical operator to be applied, which depends on the details of the syndrome extraction circuit.
For simplicity, we instead consider the distance to be the minimum distance of the stabilizer code from any ISG (which provides an upper bound on the circuit distance)~\footnote{Note that the families of Floquet codes considered in Sec.~\ref{sec:simulationresults} have a growing circuit distance which is proportional to the distance.
Such scaling of the circuit distance does not hold in general.}.

\subsection{Implementation of erasure checks and Pauli $\overline{ZZ}$ measurements}
\label{sec:ZZmeasurements}

\begin{figure}[h!]
    \centering
    \includegraphics[width=0.9\columnwidth]{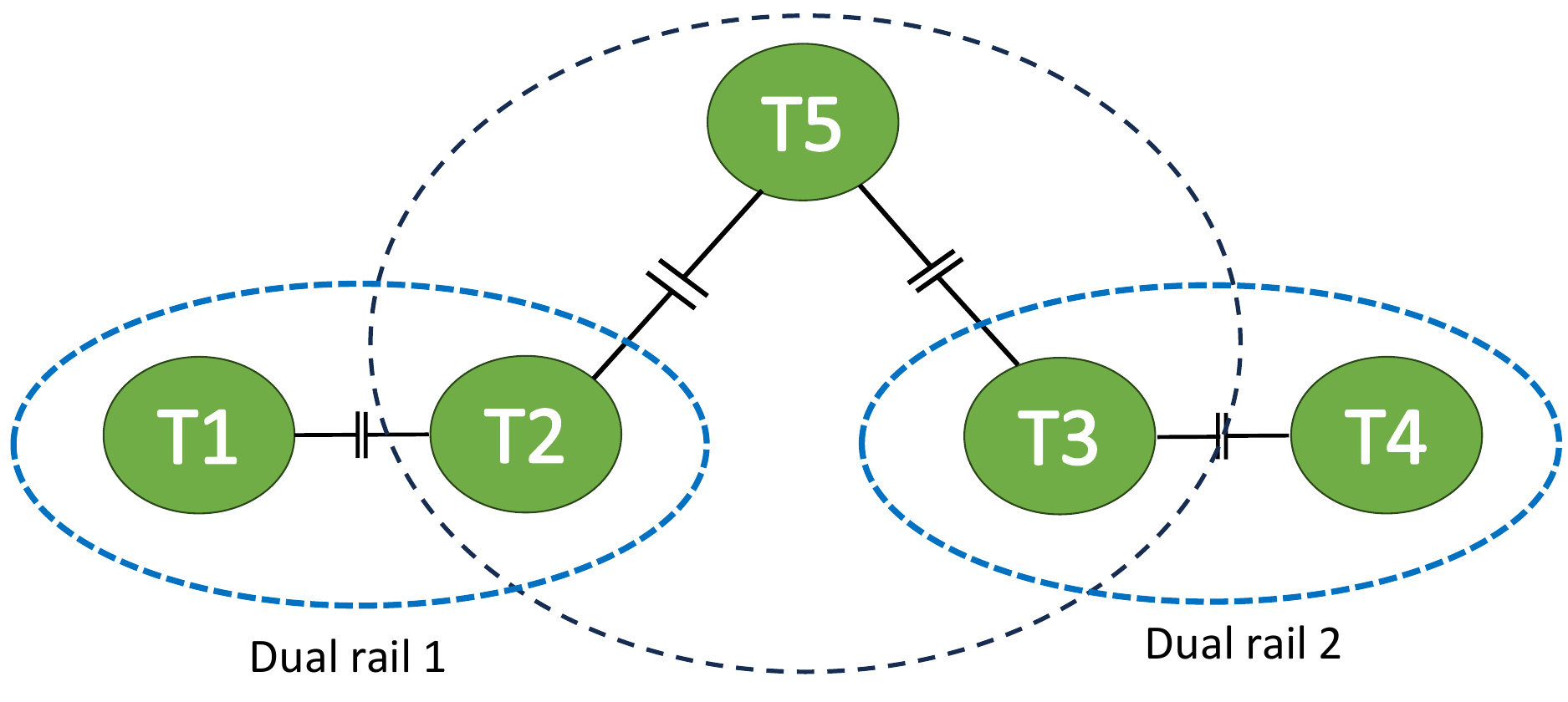}
    \caption{
    Scheme for a parity measurement of two dual-rail qubits, which are composed of transmons T1, T2 and T3, T4, respectively.
    The ancilla transmon T5 is coupled to T2 and T3.
    In this hardware-efficient construction, the utilization of a coupler is avoided by using T5 for measurement without reducing the efficiency of the procedure.
    A parity measurement could be realized by modulating the coupler energy gap at the frequency of the dual-rail qubit's gap.}
    \label{fig:DR_Yale1}
\end{figure}

Floquet codes with erasure qubits crucially rely on three operations, erasure checks, single-qubit gates on the computational subspace and projective measurements of Pauli $\overline{ZZ}$ operators.
Since an erasure check is an extra operation that is not typically considered (on top of state preparation, entangling gates and Pauli measurements), it constitutes an additional hurdle to overcome.
We mentioned that for erasure qubits via the dual-rail superconducting encoding erasure checks may, in principle, be realized by a projective measurement of the Pauli operator $ZZ$, but this is simplistic.
Instead, there are efforts to design erasure checks in an optimized way, for instance, by symmetrically coupling a readout resonator to two transmons~\cite{Kubica2023}.

Surprisingly, for dual-rail qubits projective measurements of Pauli $\overline{ZZ}$ operators might be efficiently and swiftly realized using a single transmon almost without paying the price for the transmon's low (compared to the dual-rail qubit's) coherence and amplitude damping time $T_1$.
We propose to do so by incorporating the ideas from the cavity dual-rail architecture~\cite{Teoh2023}.

\begin{figure*}[ht!]
\centering    
\includegraphics[trim={1cm 1cm 2cm 1cm},clip,width=.35\textwidth]{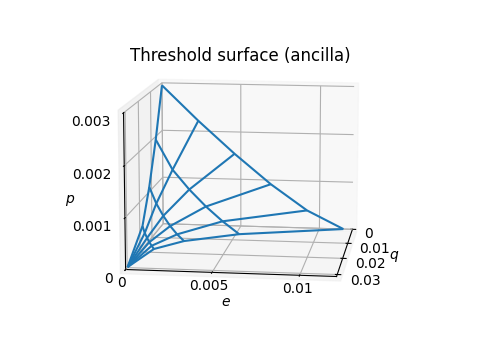}
\hspace*{-.32\textwidth}(a)\hspace*{.28\textwidth}
\includegraphics[width=.32\textwidth]{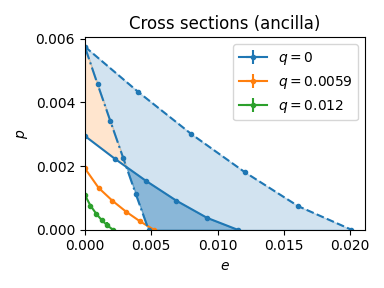}
\hspace*{-.33\textwidth}(b)\hspace*{.31\textwidth}
\includegraphics[width=.32\textwidth]{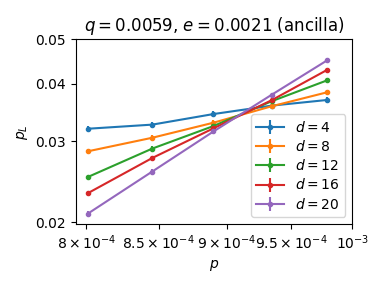}
\hspace*{-.33\textwidth}(c)\hspace*{.29\textwidth}\ \\
\includegraphics[trim={1cm 1cm 2cm 1cm},clip,width=.35\textwidth]{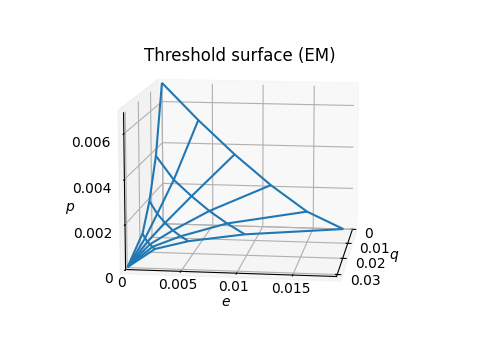}
\hspace*{-.32\textwidth}(d)\hspace*{.28\textwidth}
\includegraphics[width=.32\textwidth]{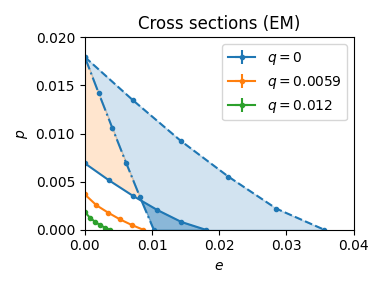}
\hspace*{-.33\textwidth}(e)\hspace*{.31\textwidth}
\includegraphics[width=.32\textwidth]{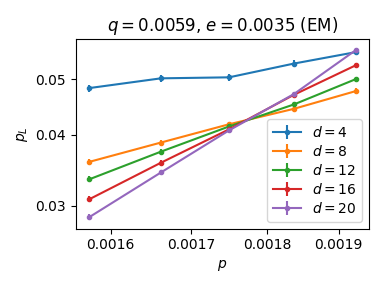}
\hspace*{-.33\textwidth}(f)\hspace*{.29\textwidth}
\caption{
Simulations of the CSS honeycomb code realized via the (a)-(c) ancilla and (d)-(f) EM schemes.
(a)(d) The threshold surface in the $(e, p, q)$ phase space, where $e$, $p$ and $q$ are the erasure, Pauli and measurement error rates, respectively.
(b)(e) Cross sections of the threshold surface for different values of $q$ (solid lines).
The dashed lines correspond to the scheme with erasure checks and reset that cause no additional errors, bounding the performance of any erasure scheme.
The dashed-dotted lines correspond to the standard scheme with no erasure checks
and ideal reset (also interpreted as the code's performance under leakage).
Erasure schemes can operate in a region (blue) where the standard scheme cannot.
Since erasure checks and reset cause additional errors, for a low erasure bias there is a region (orange), where the standard scheme may be better.
(c)(f) We find the thresholds by plotting the logical error rate $p_L$ for distance-$d$ codes as a function of $p$ or $e$, and fitting a finite-size scaling ansatz; see Appendix~\ref{app:simulationdetails}.
}
\label{fig:mainresults}
\end{figure*}

Concretely, the parity measurement could be realized by coupling a single transmon (which will be used as an ancilla) to two dual-rail qubits and modulating the flux on the transmon parametrically in resonance with the gaps of the dual-rail qubits; see Fig.~\ref{fig:DR_Yale1}.
Such a modulation realizes the following effective interaction $\frac{g_m}{2}a_5^\dagger a_5 \left(\overline{Z_1} + \overline{Z_2}\right)$, where 
$g_m$ is the interaction strength, $a_5$ is the ladder operator for the ancilla transmon and $\overline{Z_i}$ denotes Pauli $Z$ operator on the computational subspace of the $i$-th dual-rail qubit.
Since the transmon is only coupled to $\overline{Z_1} + \overline{Z_2}$, manipulating the ground and second excited states of the ancilla transmon allows for a robust parity measurement.
This method is resilient not only to the phase noise of the ancilla transmon but also to the amplitude damping noise.
Assuming that the coherence of the dual rail reaches a few milliseconds~\cite{Levine2023}, the main source of noise in this scheme is expected to be measurement idling dephasing, which should be less than $10^{-4}$; see Appendix~\ref{Implementation} for details.
By employing this scheme, we can directly implement the projective measurement of Pauli $\overline{ZZ}$ operators required for Floquet codes.

\section{Numerical simulations for Floquet codes}
\label{sec:simulationresults}

\begin{figure}[t]
\centering
(a)\includegraphics[width=.85\columnwidth]{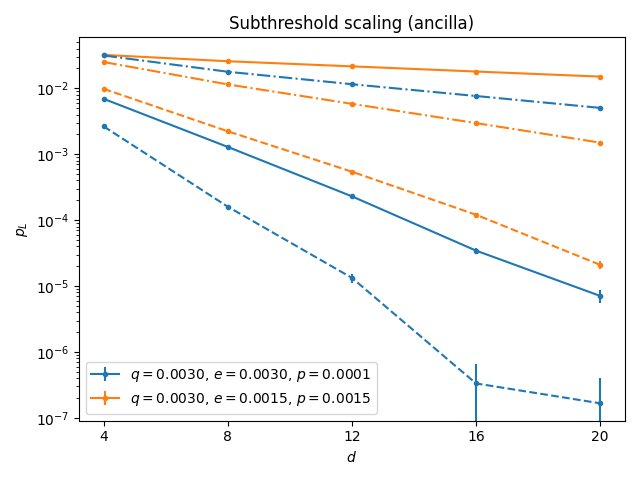}\\
(b)\includegraphics[width=.85\columnwidth]{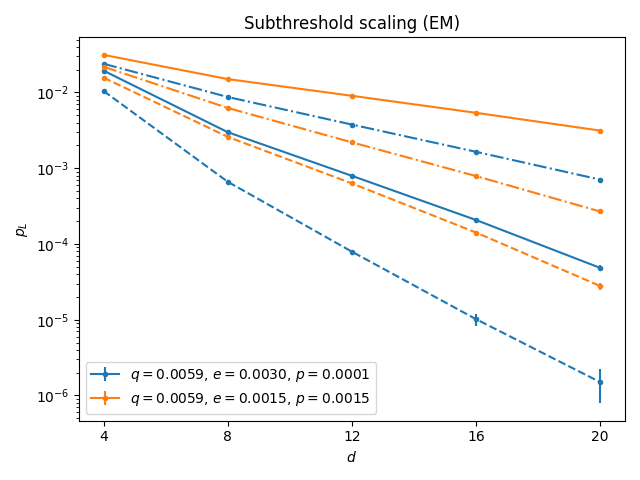}
\caption{
Subthreshold scaling of the logical error rate $p_L$ with distance $d$ for the (a) ancilla and (b) EM schemes.
We compare the results when checking for erasures after every entangling operation (solid), without performing any erasure checks (dashed-dotted), and in an ideal case where erasure checks introduce no errors (dashed), which gives an lower bound on the achievable $p_L$.
For a high erasure bias (blue), we obtain better suppression and scaling of $p_L$ by performing the erasure scheme; for a low erasure bias (orange), the standard scheme is better.}
\label{fig:subthreshold}
\end{figure}

We now describe the results of our numerical simulations of Floquet codes with erasure qubits.
Our simulations were performed using the second method of sampling described in Sec.~\ref{sec:sampling}.
After sampling erasure check detection events, we used the Python package Stim~\cite{gidney2021stim} to sample from the resulting stabilizer circuits.
For each sample, Stim outputs detectors that are violated along with the final value of a given logical operator and decomposes noise into error mechanisms that set off at most two detectors.
Thus, we decode using the method outlined in Sec.~\ref{sec:decode} by inputting this decoding graph along with the samples to the minimum-weight perfect matching decoder PyMatching~\cite{higgott2023pymatching}.
The decoder reports an error if after decoding, the value of the logical operator is different than when initialized. For a distance $d$ code, we calculate $p_L$, the logical error rate per $3d$ rounds. For more details of the simulation, see Appendix~\ref{app:simulationdetails}.

Our main numerical results are presented in Fig.~\ref{fig:mainresults}.
We simulate two ways of implementing the measurements of the CSS honeycomb code:
(i) the ancilla scheme using an ancilla qubit and two-qubit entangling gates as depicted in Fig.~\ref{fig_circuits} and
(ii) the 2Q entangling measurement (EM) scheme as described Sec.~\ref{sec:ZZmeasurements}.
In both scenarios we perform either erasure checks with reset or readout after each entangling operation.
We probe the $(e, p, q)$ phase space to determine the threshold surface
and find the correctable region where errors can be suppressed arbitrarily by increasing the code distance.

We remark that depending on the noise parameters it may be optimal to perform erasure checks less frequently than after every entangling operation.
Although we have not simulated all possible erasure check schedules, we find an upper bound for their thresholds by simulating the scheme with ideal erasure checks and reset.
In particular, the light blue and orange regions in Fig.~\ref{fig:mainresults}(b)(e) represent potential gains of the correctable region that may be achieved by optimizing the erasure check schedules.

In Fig.~\ref{fig:subthreshold}, we show how the logical error rate $p_L$ is suppressed by increasing code distance for error rates below threshold in the ancilla and EM schemes. We choose $e$, $p$ and $q$ to be in the correctable region for the erasure scheme from Fig.~\ref{fig:mainresults}.
These values are also comparable with the experimentally measured erasure and residual error rates of $0.4\%$ and $0.01\%$ per single-qubit gate, and the false positive and negative erasure detection rates of around $1\%$~\cite{Levine2023}.

\begin{figure}[t]
\centering
(a)\includegraphics[width=0.85\columnwidth]{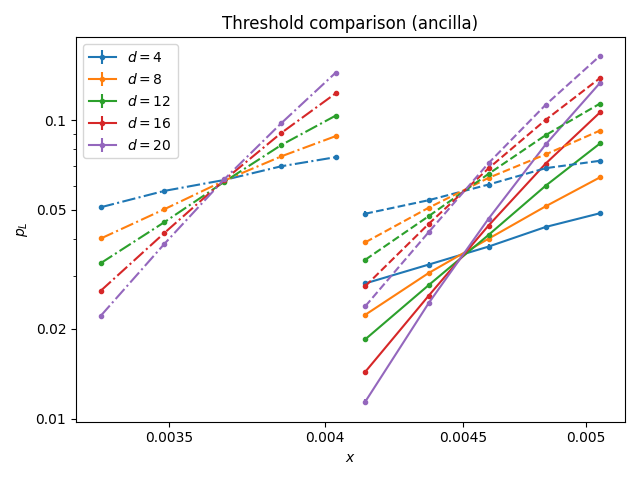}
\\
(b)\includegraphics[width=0.85\columnwidth]{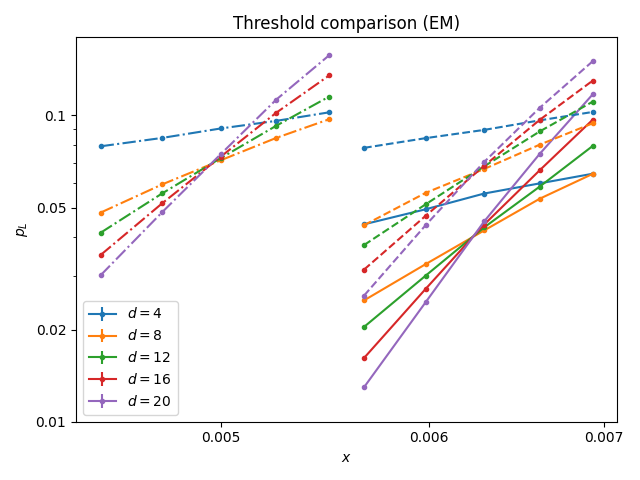}
\\
\caption{
Comparison of the thresholds for the (a) ancilla and (b) EM schemes.
We assume erasure-biased noise with a single parameter $x=q=e=10p$.
For the CSS honeycomb code, the standard (solid) and compact (dashed) layouts give the same threshold, with the latter having higher logical error rate $p_L$ for the same distance $d$.
The original honeycomb code (dashed-dotted) with the compact layout exhibits a lower threshold.
}
\label{fig:twistedlayout_numerics}
\end{figure}

In Fig.~\ref{fig:twistedlayout_numerics}, we present several optimizations where we find the threshold for the ancilla and EM schemes under erasure-biased noise characterized by a single parameter $x=q=e=10p$.
We consider two layouts: (i) the standard embedding of the hexagonal lattice on a torus as in Fig.~\ref{fig_layout_st}(a) and (ii) the qubit-efficient layout achieving the same distance by ``twisting'' the torus as in Fig.~\ref{fig_Floq_small}(f).
This qubit-efficient layout, suggested in Refs.~\cite{Gidney2021faulttolerant, Haah2022}, is the optimal layout on a torus for a given distance and uses 25\% fewer qubits than the standard layout~\cite{Bombin2007}.
Although the logical error rate $p_L$ at the threshold is lower for the standard layout, at low physical error rates, where the scaling of $p_L$ is determined by the distance, it is preferential to use the compact layout as it achieves a higher distance for a given number of physical qubits.
We also simulate the performance of the original honeycomb code with the compact layout and find that its threshold is lower than that of the CSS honeycomb code.
This can be explained by the fact that the detectors are products of 6 measurement outcomes in the CSS honeycomb code compared to products of 12 measurement outcomes in the original honeycomb code.
Therefore, the CSS version is more robust against measurement errors.

\begin{figure*}[t!]
\centering
\hspace*{-4mm}\begin{minipage}{0.8\linewidth}
(a)\includegraphics[width=0.2\textwidth]{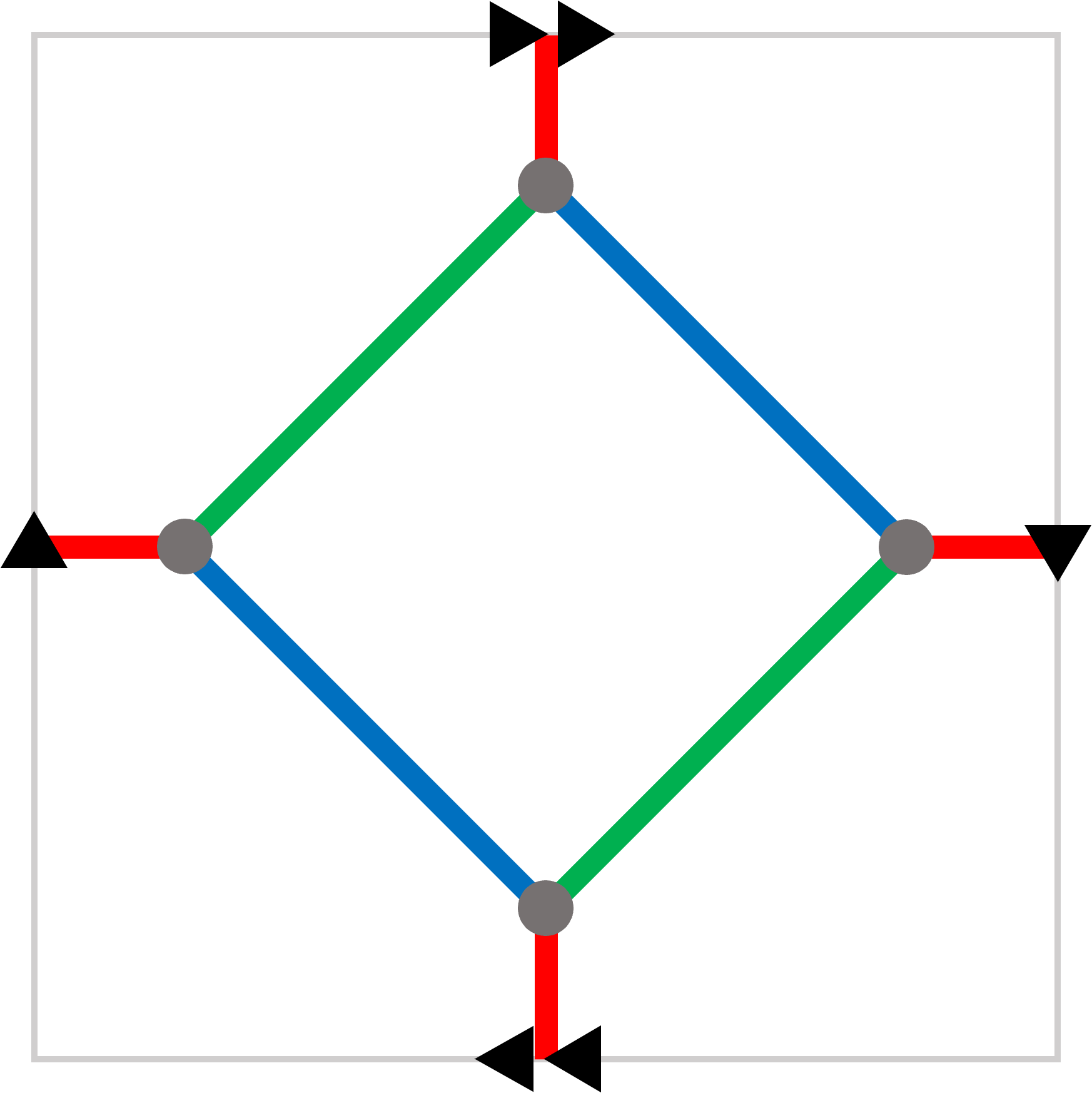}\qquad\quad
(b)\includegraphics[width=0.22\textwidth]{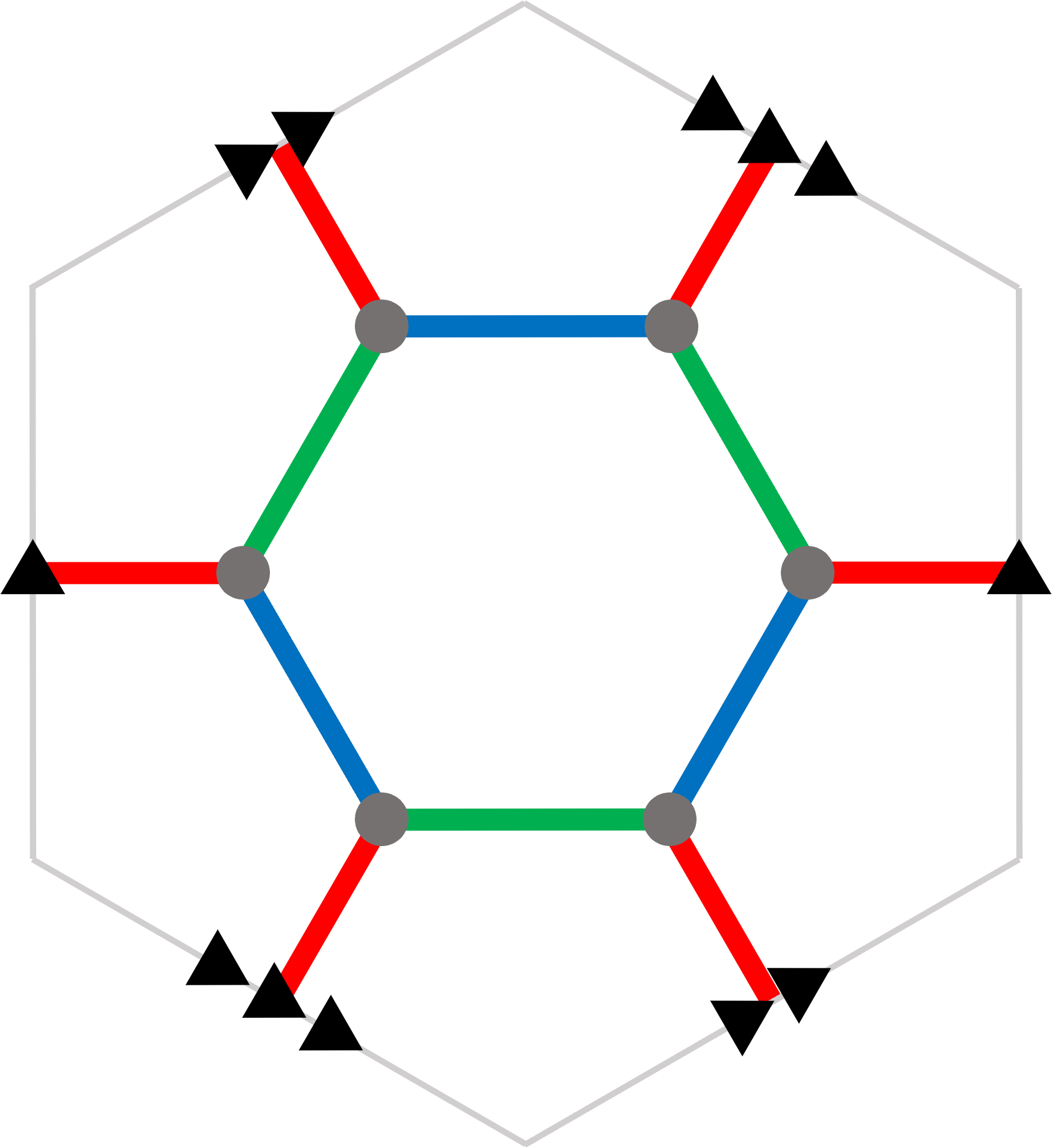}\qquad
(c)\includegraphics[width=0.22\textwidth]{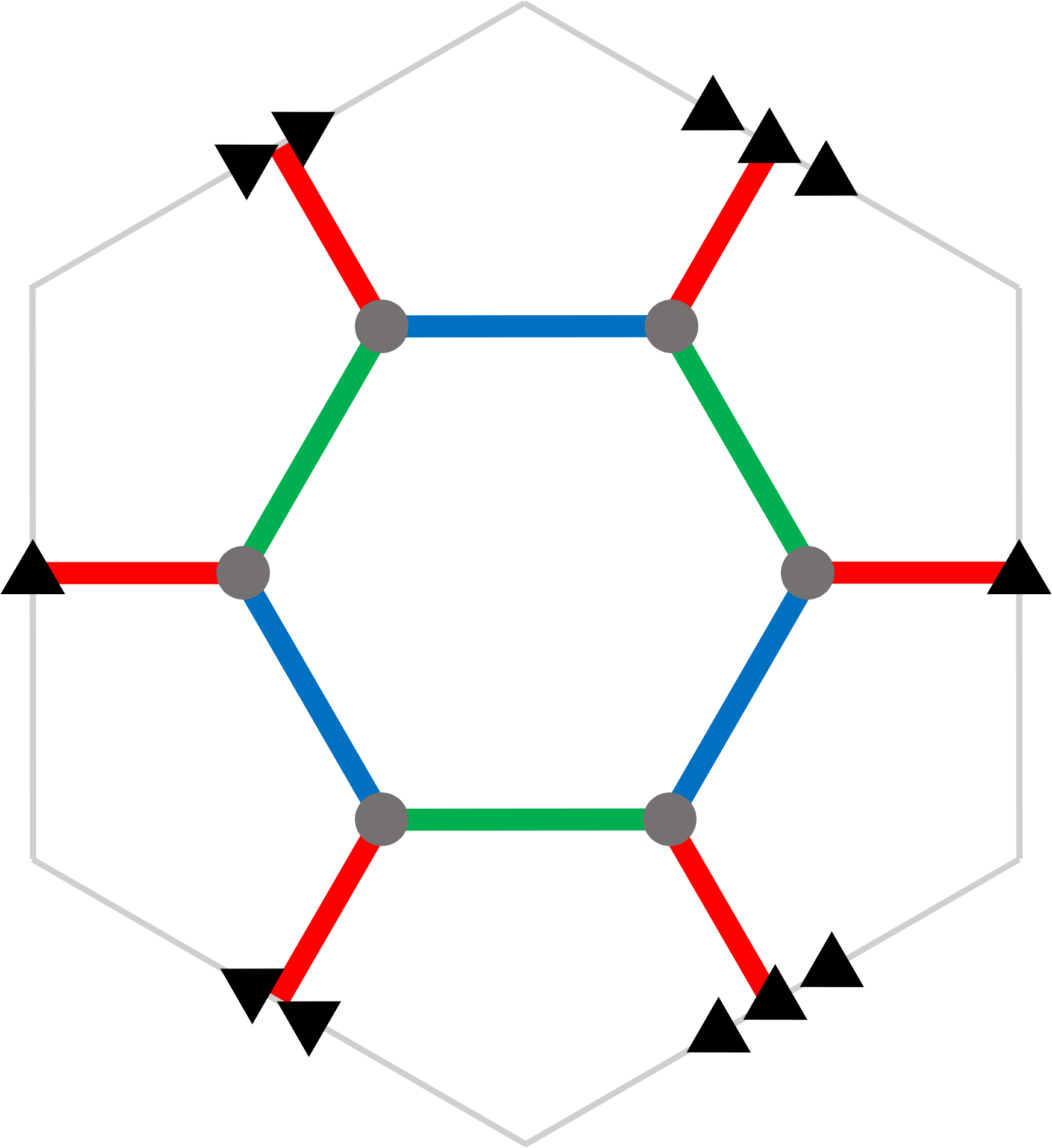}\\
(d)\includegraphics[width=0.25\textwidth]{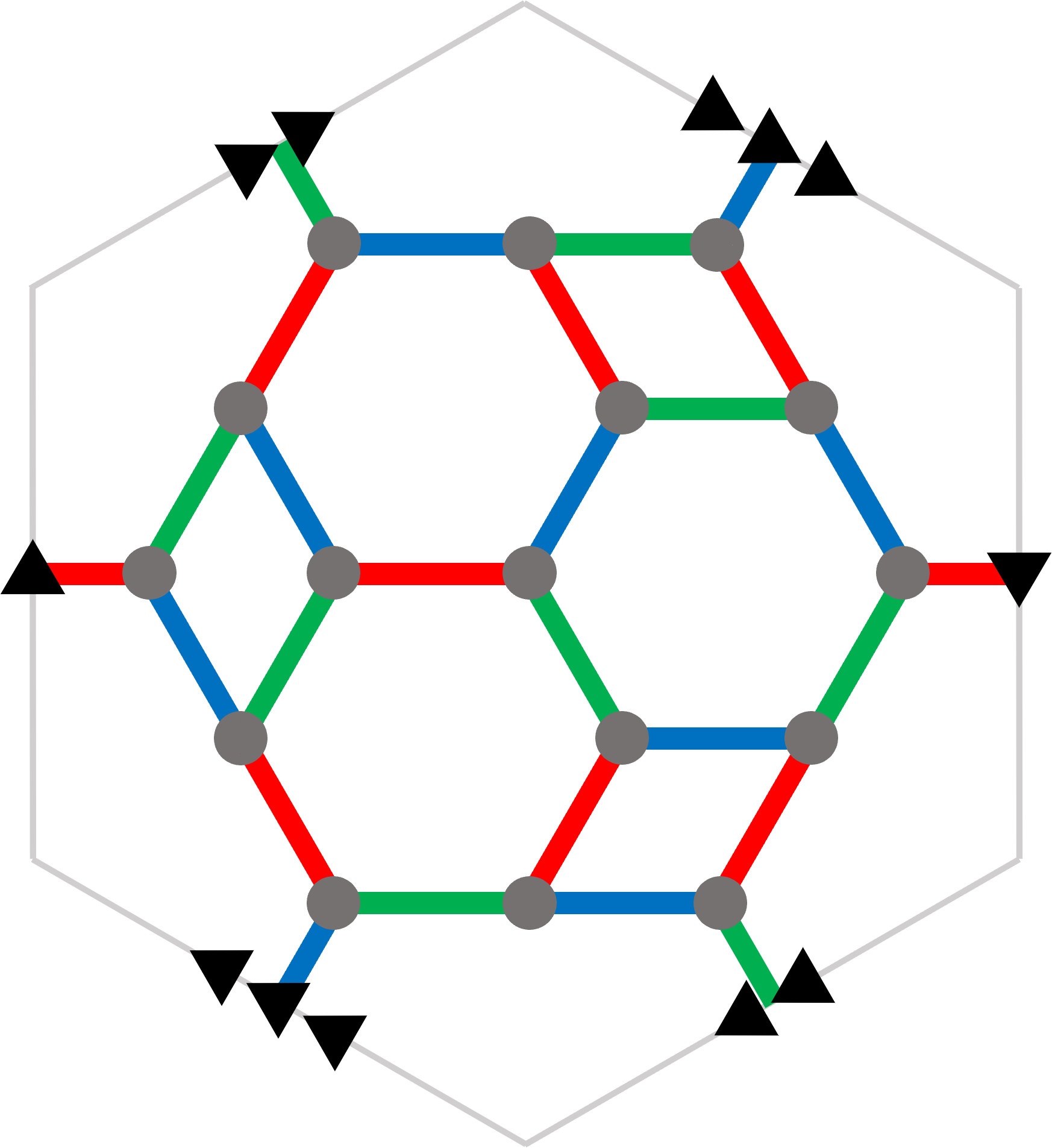}
(e)\includegraphics[width=0.28\textwidth]{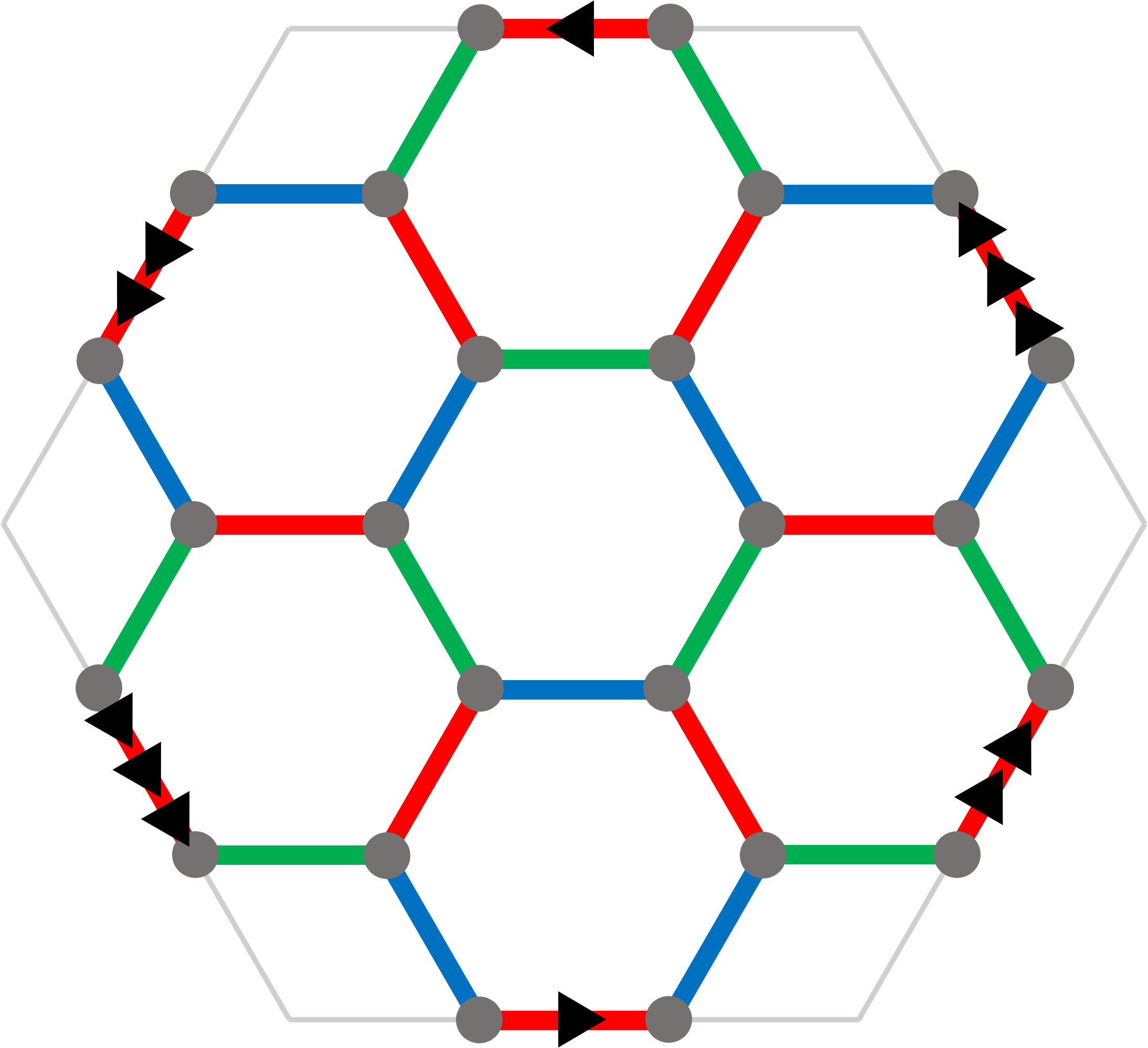}
(f)\includegraphics[width=0.21\textwidth]{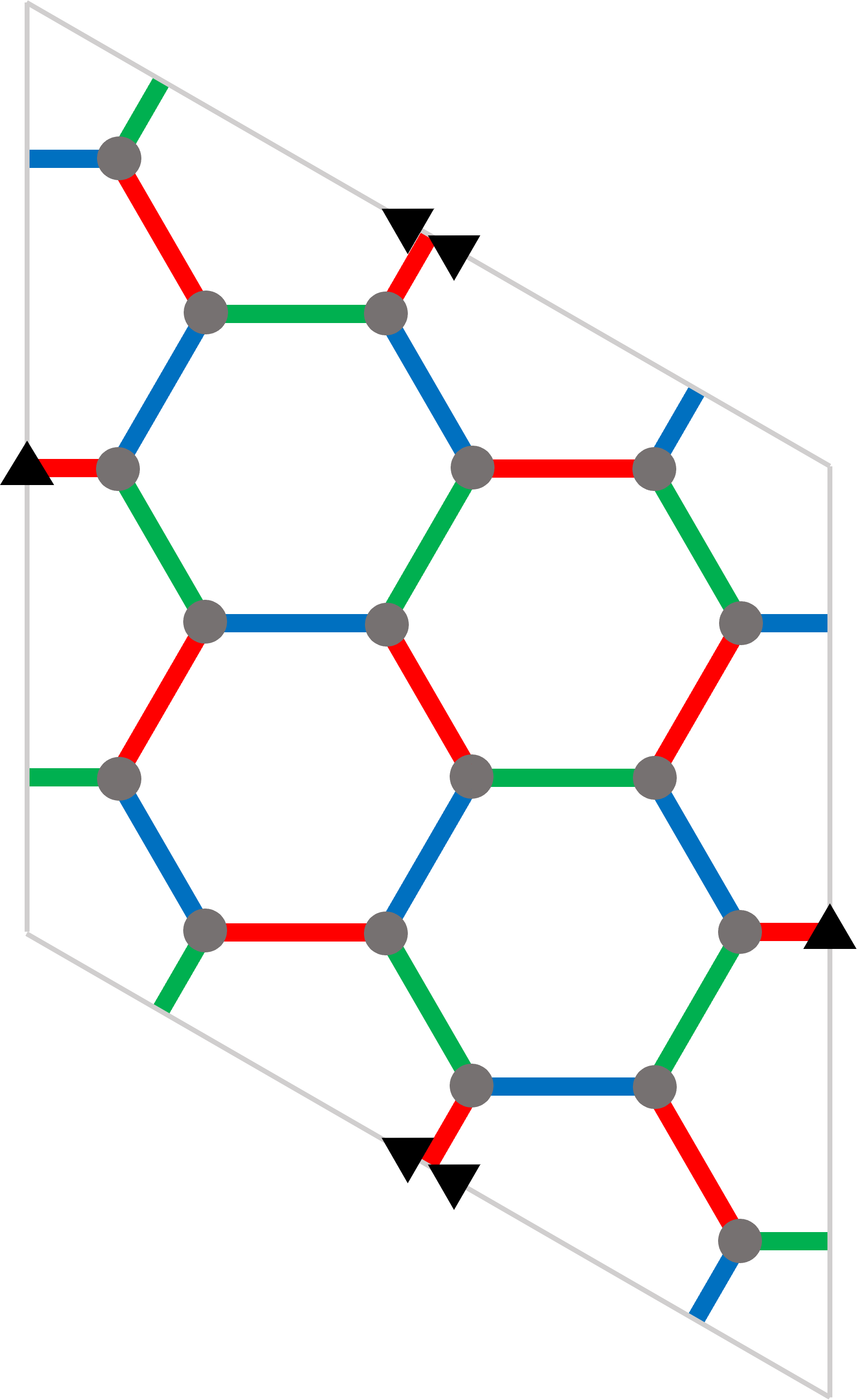}
\end{minipage}\quad
\begin{minipage}{0.2\linewidth}
\begin{tabular}{|c|c|c|c|c|}
    \hline
     {\bf code} & $\Sigma$ & $n$ & $k$ & $d$ \\
     \hline\hline 
     (a) & $\bb R\bb P^2$ & 4 & 1 & 2\\
     (b) & $\bb T^2$ & 6 & 2 & 2\\
     (c) & $\bb K$ & 6 & 2 & 2\\
     (d) & $\bb R\bb P^2$ & 16 & 1 & 4\\
     (e) & $\bb R\bb P^2$ & 18 & 1 & 4\\
     (f) & $\bb T^2$ & 18 & 2 & 4\\
     \hline
\end{tabular}
\end{minipage}
\caption{
The smallest (graph-based) Floquet codes with distance two and four.
Qubits are depicted as gray dots.
For each code illustrated in (a)-(f), we specify a manifold $\Sigma$ used to embed its associated graph, as well as its code parameters $n$, $k$ and $d$.
Here, $\mathbb T^2$, $\mathbb{RP}^2$, and $\mathbb K$ denote a torus, a real projective plane, and a Klein bottle, respectively.
}
\label{fig_Floq_small}
\end{figure*}

\section{Smallest Floquet codes}

Having analyzed families of Floquet codes on the torus, one may ask what the smallest possible (graph-based) Floquet codes are.
In this section, we find the previously unknown codes with distance two and four and analyze their performance in terms of the logical error rate.
We also describe a connection between Floquet codes and two-manifolds.

\subsection{Searching for smallest Floquet codes}
\label{sec_smallest}

Because we are considering erasures, distance-two codes may allow us to correct up to one erasure.
The smallest 3-regular graphs are the complete graph $K_4$, the complete bipartite graph $K_{3, 3}$, and the prism graph $Y_3$.
Each of these graphs has exactly one 3-edge-coloring (up to isomorphism), so they define valid Floquet codes.
The codes all have distance two, and they encode either one or two logical qubits.
We depict them in Fig.~\ref{fig_Floq_small}(a)-(c).
We remark that compared to the $\code{16, 4, 2}$ hyperbolic code defined on the Bolza surface~\cite{higgott2023hyperbolic}, the $\code{6, 2, 2}$ codes defined on $K_{3, 3}$ and $Y_3$ have better encoding rates at the same distance.

For Floquet codes that can correct one unknown error, we consider distance-four codes.
Previously, the smallest known Floquet code with distance four was the $\code{18, 2, 4}$ code using the twisted embedding of the hexagonal lattice on a torus~\cite{Gidney2021faulttolerant, Haah2022}; see Fig.~\ref{fig_Floq_small}(f).
We ran an exhaustive search through all 3-edge-colorings of 3-regular graphs up to 18 vertices and found two additional distance-four codes with 16 and 18 qubits.
These codes both encode one logical qubit; see Fig.~\ref{fig_Floq_small}(d)(e).

We also simulate the performance of the $\code{16,1,4}$ code, presenting the results for the ancilla and EM schemes in Fig.~\ref{fig:smallcode_numerics}.
To find the pseudothresholds, we compare the logical error rate $p_L$ of the code against an unprotected qubit that undergoes the same noise and is affected by four depolarizing channels, two with error rate $p$ and two with error rate $3e/4$, at every step.
In the EM scheme, there are single error mechanisms that can corrupt two qubits along a logical operator, which halves the circuit distance compared to the distance of the stabilizer code of any ISG.
This can be seen from the subthreshold scaling, as the the slopes of the solid and dashed lines are the same for low error rates.
This phenomenon does not occur for the ancilla scheme.

\begin{figure}[ht]
\centering
(a)\includegraphics[width=.85\columnwidth]{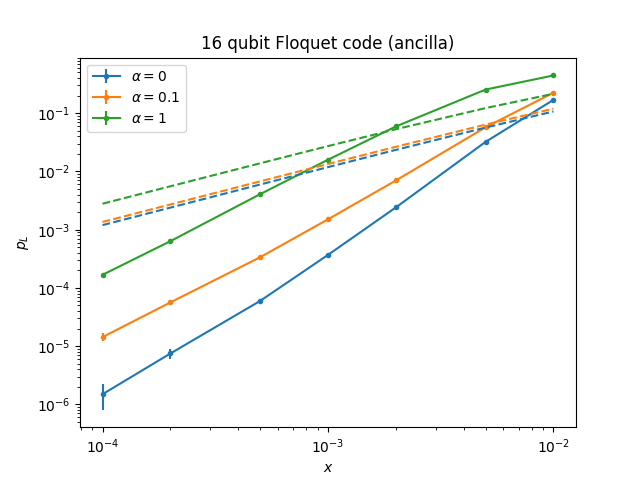}\\
(b)\includegraphics[width=.85\columnwidth]{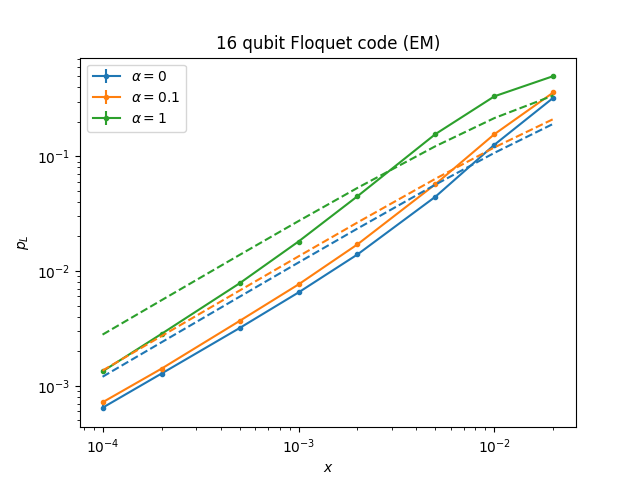}
\caption{Finding the pseudothreshold of the $\code{16, 1, 4}$ code in the (a) ancilla and (b) EM schemes. We assume erasure-biased noise with $q=e=x$ and $p=\alpha x$ where $x$ and $\alpha$ are parameters. The dashed lines represent error rates of an unprotected qubit experiencing the same noise while idling for the same amount of time.
}
\label{fig:smallcode_numerics}
\end{figure}

\subsection{Interpretation through manifolds}

It turns out that one can interpret any graph-based Floquet code as arising from a tessellation of some closed two-manifold (with the tessellation forming a two-dimensional color code lattice~\cite{Bombin2006,Kubicathesis}).
By definition, a two-dimensional color code lattice is $3$-valent and its faces are $3$-colorable, i.e., faces are colored with three colors and any two neighboring faces sharing an edge have different colors.
The following lemma guarantees the relation between Floquet codes and two-manifolds.

\begin{lemma}
    \label{prop:manifold}
    Any finite connected 3-regular graph $G = (V,E)$ with a $3$-edge-coloring $E = E_0\sqcup E_1\sqcup E_2$ can be embedded in a closed two-manifold $\Sigma$ with $3$-colorable faces, whose coloring is induced by the edge coloring.
\end{lemma}
\begin{proof}
    By removing all edges in $E_i$, we obtain a disjoint union of cycles.
    Let $F_i$ denote the collection of these cycles.
    Consider filling in these cycles so that they are homeomorphic to disks.
    The boundaries of the disks are edges of the graph, and each edge in $G$ is part of exactly two cycles.
    By gluing disks together along an edge when they share the same edge in $G$, we obtain a closed manifold $\Sigma$ on which $G$ has a natural embedding.
    The faces of $\Sigma$ are $F = F_0\sqcup F_1\sqcup F_2$ and their coloring is induced by the coloring of the edges of $G$, i.e., any face in $F_i$ has color $i$.
    Note that since each $i$-colored edge is part of faces colored $j$ and $k$, with $i,j,k$ all distinct, any two neighboring faces of $\Sigma$ have different colors (both distinct from $i$).
\end{proof}

The parameters of the Floquet code associated with the graph $G$ can be related to the $\Sigma$-embedding of $G$.
Let us define the shrunk lattice of color $i$ to be the graph $G_i = (F_i, E_i)$, where an edge $e\in E_i$ connecting $v$ to $w$ in the original graph $G$ now connects the two $i$-colored faces that $v$ and $w$ are on.
The graph $G_i$ also has an embedding in $\Sigma$, which is obtained from the embedding of $G$ by ``shrinking'' the $i$-colored faces to a point and extending the $i$-colored edges.
Similarly to Ref.~\cite{Hastings2021}, we find that at every round $i$ of the evolution of the Floquet code, the ISG is equivalent to the toric code on the $\Sigma$-embedding of $G_i$.
Thus, the number of encoded qubits is
\begin{align}
    k &= \dim  H_1(\Sigma; \bb Z_2) = 2-\chi\\
    &= \begin{cases}
		2g, & \text{orientable $\Sigma$ of genus $g$},\\
		g, & \text{nonorientable $\Sigma$ of demigenus $g$},
	\end{cases}
\end{align}
where $H_1(\Sigma; \bb Z_2)$ is the first homology group of the two-manifold $\Sigma$ with $\bb Z_2$ coefficients and $\chi$ is the Euler characteristic of $\Sigma$.
Furthermore, the distance of the Floquet code at round $i$ is the smaller of twice the length of the shortest noncontractible cycle of $G_i$ and the length of the shortest noncontractible cycle in the dual graph $G_i^*$.
Since the dual graph $G_i^*$ is bipartite, the distance is even.

\section{Discussion}

In our article, we designed and optimized fault-tolerant quantum architectures based on erasure qubits.
While our analysis has focused on Floquet codes, we also envision making use of other QEC codes, such as the surface code and quantum low-density parity-check codes~\cite{Breuckmann2021}.
The surface code, similarly to graph-based Floquet codes, can be realized with planar layouts of qubits and projective measurements of Pauli $XX$ and $ZZ$ operators between neighboring qubits~\cite{Chao2020surfacecode, Gidney2023pairmeasurement};
quantum low-density parity-check codes are generally incompatible with planar layouts, but, in principle, can be realized with, e.g., superconducting circuits~\cite{Bravyi2023} and neutral atoms~\cite{Xu2023ldpc}.
Irrespective of the QEC codes used, we expect the corresponding quantum architectures to benefit from erasure qubits and significantly outperform standard approaches.

Our analysis and numerical simulations relied on certain simplifying assumptions, including the erasure-depolarization spread, noise rates that are uniform through the circuit, and frequent erasure checks followed by unconditional reset.
However, similar analysis can be fine-tuned for specific architectures, making it more realistic and potentially further improving the performance of QEC protocols.
For instance, if erasures spread to Pauli $Z$ errors, then one may be able to design clever syndrome extraction circuits that suppress the error propagation.
One may adjust the noise rate at each spacetime location depending on the execution time of quantum circuit operations; see Appendix~\ref{app:example} for an illustrative example.
Also, one may choose to perform less frequent erasure checks (to reduce the time overhead associated with their implementation) and conditional reset operations (to reduce the effect of false negative erasure detections).

Lastly, our formalism for QEC protocols with erasure qubits and phrasing the corresponding decoding problem as the hypergraph matching problem 
constitute the first step toward systematic development and optimization of decoding algorithms.
Such efforts, in turn, will further solidify the claim that erasure qubits are an attractive building block for fault-tolerant quantum architectures.

\acknowledgements

We thank A.~Grimsmo, A.~Haim, C.~Hann, J.~Iverson and H.~Levine for many inspiring discussions.
We acknowledge C.~Pattison for his help with finding small Floquet codes.

\appendix

\section{Formal description of QEC protocols with erasure qubits}
\label{sec_formal}

We can make the discussion about QEC protocols with erasure qubits more precise. 
Formally, each wire represents an erasure qubit and it suffices to model it as a three-level system with an orthonormal basis $\ket 0$, $\ket 1$ and $\ket 2$, 
where the states $\ket 0$ and $\ket 1$ span the computational subspace $\mathcal H_c \simeq \mathbb C^{\otimes 2}$ and the state $\ket 2$ spans the erasure subspace $\mathcal H_e \simeq \mathbb C$.
Let $\Pi_a$ and $\Pi_{a,b}$, where $a,b\in \{c,e\}$, denote the projectors onto $\mc H_a$ and $\mc H_a\otimes \mc H_b$,  respectively.
Similarly, we define $\Pi^\pm_P$ and $\Pi^\pm_{PP}$, where $P\in\{X,Y,Z\}$, to be the projectors onto the $(\pm 1)$-eigenspaces of the Pauli $P$ and $PP$ operators, respectively.
We write $G_{a,b}$ to capture that the operator $G$ acts on $\mc H_a \otimes \mc H_b$.

In Sec.~\ref{sec_formalism} we assumed that none of the operations (i)-(vii) can create a superposition of states in the computational and erasure subspaces of erasure qubits.
Given our assumption of the erasure-depolarization spread, i.e., an erasure causes full depolarization of other qubit that is involved in the same 2Q operation, we obtain that the operations (i)-(vii) have a block-diagonal structure and act on the Hilbert spaces associated with erasure qubits as follows.
\begin{itemize}
\item[(i)] 1Q state preparation of a state $\ket\psi\in\mathcal H_c$ in the computational subspace of the erasure qubit.
\item[(ii)] 1Q readout measures a Pauli $P$ operator, but if the state is erased, then it gives a random outcome, i.e., it performs the two-outcome positive operator-valued measure (POVM) with $\Pi_P^+ + \frac 1 2\Pi_e$ and $\Pi_P^- + \frac 1 2\Pi_e$.
\item[(iii)] 1Q gate $G$ acts on the computational subspace of the erasure qubit, i.e., $G_{c} \oplus I_{e}$.
\item[(iv)] 2Q gate $G$ acts on the computational subspace of the two erasure qubits and fully depolarizes the other qubit if one qubit is erased, i.e., it applies a quantum channel with Kraus operators
$K_{P,Q} = \frac 1 4 G_{c,c} \oplus P_{c,e} \oplus Q_{e,c} \oplus I_{e,e}$
for all $P, Q\in \{I,X,Y,Z\}$.
\item[(v)] 1Q erasure check performs the two-outcome measurement with projectors $\Pi_c$ and $ \Pi_e$.
\item[(vi)] 1Q reset acts trivially on the computational subspace and reinitializes an erased state as the maximally mixed state in the computational subspace, i.e., it applies a quantum channel with Kraus operators $K_0 = \Pi_c$, $K_1 = \tfrac{1}{\sqrt{2}}\ketbra{0}{2}$ and $K_2 = \tfrac{1}{\sqrt{2}}\ketbra{1}{2}$.
\item[(vii)] 2Q projective measurement measures a Pauli $PP$ operator, but if either qubit is erased, then it gives a random outcome and fully depolarizes the other qubit, i.e., it performs the two-outcome POVM with
$\Pi_{PP}^+ + \frac 1 2(\Pi_{e,c} + \Pi_{c,e} + \Pi_{e,e})$ and  $\Pi_{PP}^- + \frac 1 2(\Pi_{e,c} + \Pi_{c,e} + \Pi_{e,e})$ followed by an application of a quantum channel with Kraus operators
$K_{P,Q}= \frac 1 4 I_{c,c} \oplus P_{c,e} \oplus Q_{e,c} \oplus I_{e,e}$
for all $P, Q\in \{I,X,Y,Z\}$.
\end{itemize}

\section{Examples of the edge-weight calculation and erasure rate adjustment}
\label{app:example}

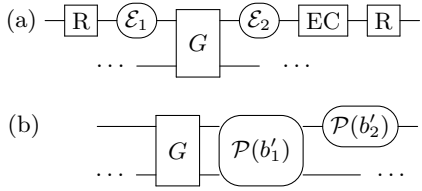
\begin{figure}[ht!]
    \centering
    (a) $\Qcircuit @C=.8em @R=.7em {
        & \gate{\mathrm{R}} & \measure{\mc E_1} & \multigate{1}{G} & \measure{\mc E_2} & \gate{\mathrm{EC}} & \gate{\mathrm R} & \qw\\
        & & \lstick{\ldots} & \ghost{G} & \qw & \lstick{\ldots}}$\\
    \vspace*{10pt}
    (b) \qquad $\Qcircuit @C=.8em @R=.7em {
        & \qw & \qw & \multigate{1}{G} &  \multimeasure{1}{\mc P(b'_1)} & \measure{\mc P(b'_2)} & \qw\\
        & & \lstick{\ldots} & \ghost{G} & \ghost{\mc P(b_1)} & \qw & \lstick{\ldots}}$ 
    \caption{
    (a) An example of a segment of an erasure circuit.
    (b) An equivalent stabilizer circuit.
    }
    \label{fig:exampleerasurecircuit}
\end{figure}

We present an example of how to decode erasures by converting a segment of an erasure circuit into a stabilizer circuit.
Consider the segment in Fig.~\ref{fig:exampleerasurecircuit}(a), where each erasure location has probability $e$ and the erasure check outcome is flipped with probability $q$.
The erasure detection event $\mathrm{EC}$ has distribution
\begin{align}
    \Pr(\mathrm{EC}=1) &= [1 - (1 - e)^2](1 - q) + (1 - e)^2q\, ,\\
    \Pr(\mathrm{EC}=0) &= 1 - \Pr(\mathrm{EC}=1)\, .
\end{align}
Conditioned on $\mathrm{EC}$, the probabilities that the qubit was first erased at $\mc E_1$ or $\mc E_2$ are respectively
\begin{align}
        a_1 = \frac{e(1-q)}{\Pr(\mathrm{EC}=1)},\qquad a_2 = \frac{(1-e)e(1-q)}{\Pr(\mathrm{EC}=1)}\, ,
\end{align}
if $D=1$, and
\begin{align}
    a_1 &= \frac{eq}{\Pr(\mathrm{EC}=0)},\qquad a_2 = \frac{(1-e)eq}{\Pr(\mathrm{EC}=0)}\, ,
\end{align}
if $\mathrm{EC}=0$. If the qubit was first erased at $\mc E_1$, it would depolarize the second qubit after the entangling gate. Furthermore, the qubit itself would become maximally mixed after the reset operation. Thus, both qubits become fully depolarized. If the qubit was first erased at $\mc E_2$, only that qubit would become depolarized from the reset.
By the proof of Lemma~\ref{thm:segmentmapping}, the segment is equivalent to the stabilizer circuit in Fig.~\ref{fig:exampleerasurecircuit}(b) with error probabilities
\begin{align}
    b'_1 = \frac{15}{16} a_1,\qquad b'_2 = \frac{3}{4} \frac{a_2}{1-a_1}\, .
\end{align}

To adjust erasure rates at different spacetime locations depending on the execution time of quantum circuit operations (and thereby making numerical simulations more realistic) we can use the following simple heuristic.
Let $T_E$ be the erasure time (which for the erasure qubit via the dual-rail encoding corresponds to the amplitude damping time $T_1$). 
Let $\mc E$ be an erasure location in between two consecutive quantum operations $A$ and $B$ with the execution time $T_A$ and $T_B$, respectively.
We can then set the erasure rate associated with $\mc E$ to be
\begin{equation}
e = (a T_A + b T_B)/T_E,
\end{equation}
where $a,b\in[0,1]$ are appropriately chosen.
In particular, in the middle of the segment $s$ we may set $a=b=0.5$; if $A$ or $B$ correspond to one of the endpoint of $s$, then we set $a$ or $b$ to be 1.
We also remark that adjusting $a$ and $b$ for erasure locations adjacent to erasure checks allows us to effectively adjust the false positive and negative erasure detection rates.

\section{Parity measurement of two dual-rail qubits}
\label{Implementation}

Here, we outline our scheme for parity measurement of two dual-rail qubits that utilizes a single transmon for measurement.
The proposed scheme is based on Fig.~\ref{fig:DR_Yale1}.
The dual-rail qubits are encoded in transmons T1, T2 and T3, T4, respectively, while the interaction is generated by the coupler T5 within the T2, T5, T3 system.

Each dual-rail qubit consists of two tunable transmons, brought to resonance as in Refs.~\cite{Kubica2023,Levine2023}, while the single tunable transmon T5 facilitates the parity measurement.
The reason why we can employ such a hybrid construction combining high-coherence dual-rail qubits and a low-coherence transmon is that most of the noise on the ancilla transmon commutes with the interaction, and thereby does not propagate in leading order to the dual-rail qubits as in Ref.~\cite{zuk2023robust}.
Part of the noise that does propagate is addressed by the dual-rail qubit's built-in decoupling mechanism.

Furthermore, since this construction is similar in nature to the cavity setup~\cite{Teoh2023}, the interaction can be made much faster as it is not limited by the Purcell effect, which is the main limitation of the rate for high-coherence cavities in a hybrid construction.

The interaction is generated by the second order $ZZ$ coupling between T5 and T1 and T5 and T3 in the following way.
Starting with the Hamiltonian
\begin{eqnarray}
H &=& \sum_{i=1}^5 \omega_i a_i^\dagger a_i + \frac{\alpha}{2} a_i^\dagger a_i^\dagger a_i a_i +\\ \nonumber
&&g_{DR1} \left(a_1^\dagger a_2  + h.c. \right)+
g_{DR2} \left( a_3^\dagger a_4 + h.c. \right)+\\\nonumber
&&g \left( a_2^\dagger a_5 +a_5^\dagger a_3  + h.c. \right),
\end{eqnarray}
where $\omega_i$ is the frequency of the transmon T$i$, $\alpha$ is the nonlinearity, $g_{DRi}$ is the capacitive coupling between two transmons of the $i$-th dual-rail qubit and $g$ is the capacitive coupling between the ancilla transmon T5 and either T2 or T3.

In the limit of large detunings, $\Delta_1=\omega_5-\omega_2$, $\Delta_2=\omega_5-\omega_3$, when $\omega_2=\omega_1$, $\omega_4=\omega_3$ and $\Delta_{i} \gg \alpha,g,$ the coupling between T5 and the rest of the system reduces to 
\begin{equation}
g_{zz}^{(1)} a_5^\dagger a_5 a_2^\dagger a_2 + g_{zz}^{(2)} a_5^\dagger a_5 a_3^\dagger a_3, 
\end{equation}
when $g_{zz}^{i} = \frac{g^2}{\Delta_{i}^2}\alpha.$ 
This Hamiltonian becomes
\begin{equation}
H = \frac{g_{DR1}}{2}  \sigma_z^1  + \frac{g_{DR2}}{2} \sigma_z^2   +  a_5^\dagger a_5 \left(g_{zz
}^{(1)} \sigma_x^1 + g_{zz
}^{(2)} \sigma_x^2 \right).
\end{equation}
Here, $\sigma_z^{i}$ and $\sigma_x^{i}$ are Pauli operators of the $i$-th dual-rail qubit defined in the standard way for the computational basis states $\ket 0$ and $\ket 1$, where
$\ket{b} = \ket{ g e - (-1)^b e g}$, with $\ket g$ and $\ket e$ denoting the transmon's ground and first excited states, respectively (note the change of basis compared to the main text).
In principle the interaction terms are off resonance and thus could be neglected unless $g_{zz
}^{(i)}$ is modulated at the dual-rail qubit frequency $g_{DRi}$.

The $ZZ$ coupling term $g_{zz}^{(i)}$ could be modulated by a parametric drive of the detuning resulting in $g_{zz}^{i} = \frac{g^2}{(\Delta_{i}+\delta \cos \Omega t)^2}\alpha \approx  \frac{g^2}{\Delta_{i}^2}\alpha  -2 \frac{g^2}{\Delta_{i}^2}\alpha \frac{\delta}{\Delta_{i}}\cos \Omega t$.
Thus, by modulating the gap detuning for one of the dual-rail qubits, it is possible to realize an effective Hamiltonian $\frac{g_m}{2} a_5^\dagger a_5 \sigma_x^1$ or $\frac{g_m}{2} a_5^\dagger a_5 \sigma_x^2$. 
Moreover, by modulating at both frequencies it is possible to realize a Hamiltonian $\frac{g_m}{2} a_5^\dagger a_5 (\sigma_x^1 + \sigma_x^2)$. 
By adding local terms this Hamiltonian could be written as $H_p = \frac{g_m}{2} \vert f \rangle \langle f \vert (\sigma_x^1 + \sigma_x^2)$, where $\ket f$ denotes the transmon's second excited state.
This is exactly what is needed to implement the $ZZ$ gate proposed in Ref.~\cite{Teoh2023}.
Here, however, we only aim to use this term to implement a parity measurement, simplifying the scheme and resulting in higher fidelity.

\subsection{Parity measurement scheme}

We propose to conduct the parity measurement via the ground state $\ket g$ and the second excited state $\ket f$ manifold, enabling us to detect amplitude damping of the transmon by measuring the first excited state $\ket e$.
The protocol starts with the state
\begin{eqnarray}
\vert g + f\rangle \vert \alpha 00  + \beta 10 + \gamma 01 + \delta 11\rangle,    
\end{eqnarray}
where $\ket{b} = \ket{ g e - (-1)^b e g}$ denotes the computational basis state of the dual-rail qubit. 
By setting the total time so that $\frac{g}{2} t = \pi$ the unitary $e^{i H_p t}$ propagates the state to
\begin{equation}
   \vert g + f\rangle \vert \alpha 00   + \delta 11\rangle + \vert g - f\rangle \vert  \beta 10 + \gamma 01 \rangle, 
\end{equation}
which, followed by measuring the operator $\ket{g+f}\bra{g+f}-\ket {g-f}\bra{g-f}$, realizes the parity measurement.
The main advantage of this scheme is that amplitude damping is heralded, and the phase noise of the transmon does not propagate to the dual-rail qubits because the coupling term $a_5^\dagger a_5$ commutes with the noise term.
Thus, the transmon's phase noise only affects the measurement error, benefiting from a large threshold, as discussed in Ref.~\cite{zuk2023robust}.
This property of the measurement scheme allows for the utilization of a low-coherence single transmon and does not necessitate the use of a dual-rail qubit as an ancilla or a coupler, which would have slowed down the protocol considerably.

\begin{figure}[t!]
    \centering
    \includegraphics[width=0.85\columnwidth]{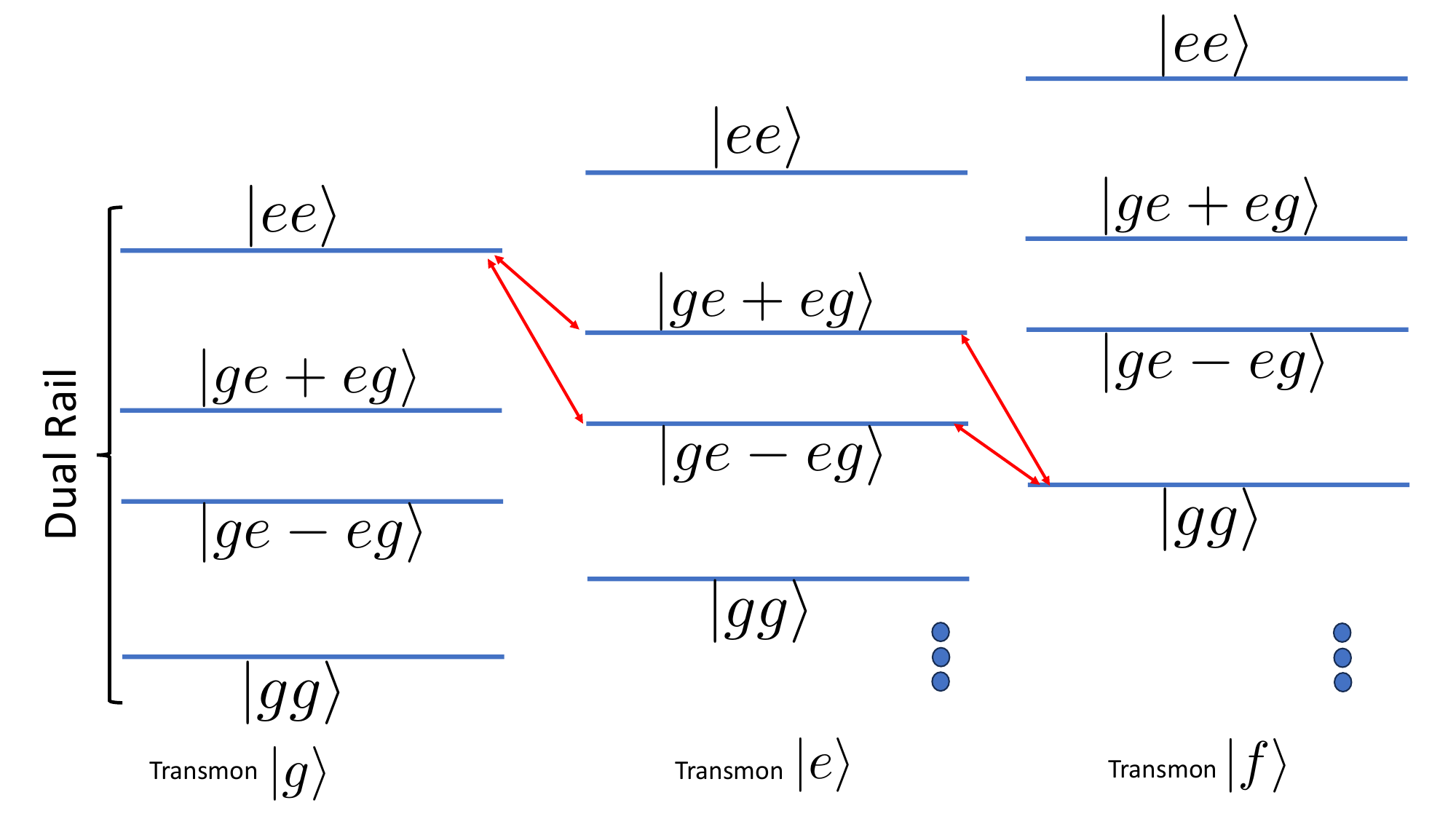}
    \caption{Error terms that limit the fidelity of the parity measurements are shown in red.}
    \label{fig:DR_Yale6}
\end{figure}

The parity measurement itself is slower by a factor of 2 than the regular $ZZ$ term between two transmons and taking into account that the gate should be realized twice, once for each dual rail, the gate could be achieved at around $50\unit{ns}$.
The dephasing due to the dual-rail qubit would be of the order of $\sim \frac{50\unit{ns}}{2.5\unit{ms}} = 2 \times 10^{-5}$ for the Markovian case and down to the $10^{-10}$ level for the non Markovian one, where the estimation of the  $2.5 \unit{ms}$ coherence time is based on results in Ref.~\cite{Levine2023}.
Thus, the main source of dephasing would be the erasure measurements.
The erasure measurement takes a similar amount of time as a regular measurement, which could be performed in less than $100\unit{ns}$~\cite{heinsoo2018rapid} or even below $50\unit{ns}$~\cite{sunada2023photon}. 
The measurement idiling dephasing will be the main source of noise in this scheme and should be less than $10^{-4}$ assuming that the coherence of the dual-rail qubit would reach a few miliseconds~\cite{Levine2023}.

The low transmon's coherence will only limit the measurement fidelity to the level of $\sim \frac{50\unit{ns}}{T_2} \sim 0.2 \%$, where $50\unit{ns}$ is the parity measurement time, which is not a real limitation as measurement fidelities are already limited to that level due to other reasons.

The gate speed will be limited by the off-resonant error terms which will cause leakage at the end of the gate.
These transitions are shown in Fig.~\ref{fig:DR_Yale6}.
As these transitions are detuned by $\frac{\omega_T - \omega_{DR}}{2}$  the error terms are approximately  $\frac{8}{((\omega_T - \omega_{DR}) T_g)^4}$~\cite{boradjiev2013control}.
Assuming we target a fidelity of around $10^{-4}$ the gate time should be of the order of $40\unit{ns}$.

\begin{figure}[t!]
    \centering
    (a)\includegraphics[width=.85\columnwidth]{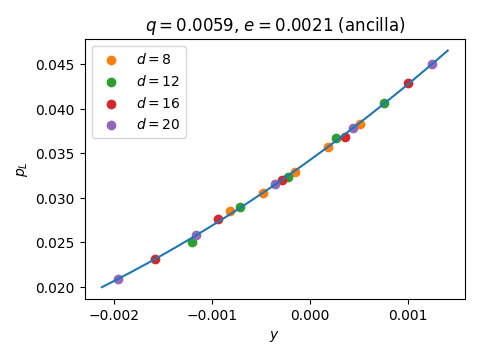}\\
    (b)\includegraphics[width=.85\columnwidth]{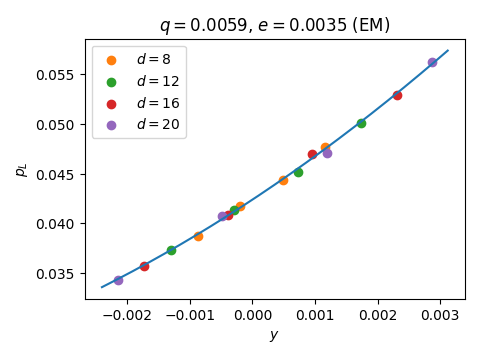}
    \caption{Rescaled data for the sample threshold calculations in Fig.~\ref{fig:mainresults} based on Eqs.~\eqref{eq:pLquadratic}~and~\eqref{eq:rescaledvariable}, where $y=p$ is the swept variable. (a) Sample calculation for the ancilla scheme threshold, giving $y^*=9.1\times 10^{-4}$, $\alpha=0.97$, and the quadratic $p_L = 570y^2 + 7.9y + 0.034$. (b) Sample calculation for the EM scheme threshold, giving $y^*=1.8\times 10^{-3}$, $\alpha=0.99$, and the quadratic $p_L = 210y^2 + 4.2y + 0.042$}
    \label{fig:universalansatz}
\end{figure}

\begin{figure}[t!]
\centering
(a)\includegraphics[width=.85\columnwidth]{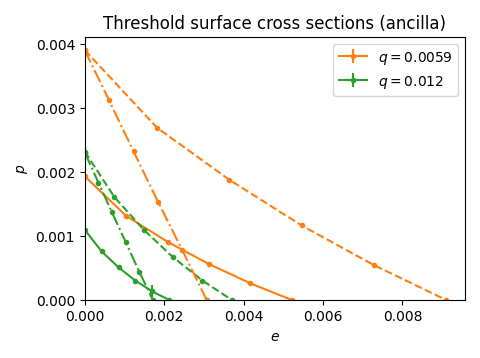}\\
(b)\includegraphics[width=.85\columnwidth]{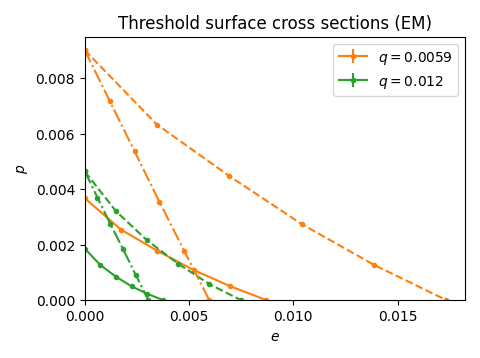}
\caption{Cross sections of the threshold surfaces from Fig.~\ref{fig:mainresults} for different values of $q$ (solid lines) for the (a) ancilla and (b) EM schemes.
The dashed lines correspond to the scheme with erasure checks and reset that do not introduce additional errors, bounding the performance of any erasure scheme.
The dashed-dotted lines correspond to the standard scheme with no erasure checks
and ideal reset (which can also be interpreted as the code's performance under leakage).
}
\label{fig:crosssectionssupplementary}
\end{figure}

\section{Details of numerical simulations}
\label{app:simulationdetails}

We present more details on how the simulations were performed. For a given circuit, the state is initialized as the eigenstate of a chosen logical operator, and an error is reported if the logical operator is decoded to the wrong value at the end of the simulation~\footnote{Because Floquet codes encode two logical qubits, the word error rate is four times the values presented if we assume independent $X$ and $Z$ failure probabilities. The thresholds will remain the same.}.
Because we are interested in threshold values, we assume perfect initialization and a noiseless final measurement. We run the simulation for $9d$ noisy measurement rounds for $d=4, 8, 12, 16, 20$ to obtain $p_L'$, and then report the normalized error rate per $3d$ rounds calculated via $p_L = \frac 1 2(1 - (1 - 2p_L')^{1/3})\approx p_L'/3$.
The logical error rate $p_L$ is calculated as the average over at least 1000 circuit realizations (from a given pattern of erasure check detection events), where each circuit realization is sampled 200 times.

The threshold surfaces in Fig.~\ref{fig:mainresults} are obtained by sweeping an error parameter (usually $p$, but sometimes $q$ or $e$ for points where $p=0$) in the neighborhood of a suspected threshold point in the $(e, p, q)$ phase space. The threshold value is estimated by fitting the universal scaling ansatz for critical points of phase transitions~\cite{Wang2003,Harrington2004}. That is, around the threshold, we assume the form
\begin{equation}
    \label{eq:pLquadratic}
    p_L = ax^2 + bx + c
\end{equation}
for the scaled variable
\begin{equation}
    \label{eq:rescaledvariable}
    x = (y - y^*)d^{\alpha}\, ,
\end{equation}
where $p_L$ is the logical error rate, $y\in\{e, p, q\}$ is the swept error variable, and $a, b, c, y^*, \alpha$ are fitting parameters; see Fig.~\ref{fig:universalansatz} for example calculations.
In Fig.~\ref{fig:crosssectionssupplementary}, we present additional cross sections of the threshold surfaces from Fig.~\ref{fig:mainresults}.

\bibliography{bib_dualfloq}

\end{document}